\documentclass[acmsmall,screen,11pt]{acmart}
\usepackage[reqno]{amsmath}
\usepackage[inference]{semantic}
\usepackage{verbatim}
\usepackage{stmaryrd}

\usepackage{amssymb} 
\let\originalo\o
\renewcommand{\o}{\text{\originalo}}
\renewcommand{\O}{\varnothing}
\usepackage{amsthm} % If you want to use AMS-LaTeX, comment out these
\usepackage{bbm}
\usepackage{pifont}
\usepackage{mathdots}
\usepackage{graphicx}
\usepackage{color}

\newtheorem{theorem}{Theorem}
\newtheorem{proposition}[theorem]{Proposition}
\newtheorem{lemma}[theorem]{Lemma}
\newtheorem{corollary}[theorem]{Corollary}
\newtheorem{definition}[theorem]{Definition}

\setcopyright{none}
\renewcommand\footnotetextcopyrightpermission[1]{}

\begin{document}

%%
%% The "title" command has an optional parameter,
%% allowing the author to define a "short title" to be used in page headers.
\title{Reachability in Geometrically $d$-Dimensional VASS}
\author{Yuxi Fu}
    \email{fu-yx@cs.sjtu.edu.cn}
    \affiliation{
        \institution{BASICS, Shanghai Jiaotong University}
        \city{Shanghai}
        \country{China}
    } 
\author{Yangluo Zheng}
    \email{wunschunreif@sjtu.edu.cn}
    \affiliation{
        \institution{BASICS, Shanghai Jiaotong University}
        \city{Shanghai}
        \country{China}
    } 
\author{Qizhe Yang}
    \email{qzyang@shnu.edu.cn}
    \authornote{Corresponding author.}
    \affiliation{
        \institution{Shanghai Normal University}
        \city{Shanghai}
        \country{China}
    }

%%
%% The abstract is a short summary of the work to be presented in the
%% article.
\begin{abstract}
Reachability of vector addition systems with states (VASS) is Ackermann complete~\cite{leroux2021reachability,czerwinski2021reachability}.
For $d$-dimensional VASS reachability it is known that the problem is NP-complete~\cite{HaaseKreutzerOuaknineWorrell2009} when $d=1$, PSPACE-complete~\cite{BlondinFinkelGoellerHaaseMcKenzie2015} when $d=2$, and in $\mathbf{F}_d$~\cite{FuYangZheng2024} when $d>2$.
A geometrically $d$-dimensional VASS is a $D$-dimensional VASS for some $D\ge d$ such that the space spanned by the displacements of the circular paths admitted in the $D$-dimensional VASS is $d$-dimensional.
It is proved that the $\mathbf{F}_d$ upper bounds remain valid for the reachability problem in the geometrically $d$-dimensional VASSes with $d>2$.
\end{abstract}





\maketitle
\section{Introduction}

Petri nets offer a model for concurrent systems in which system parameters are represented by a vector of nonnegative integers~\cite{Peterson1981}.
In a particular state a Petri net may enable a particular system transition that updates the system parameters.
The parameter updates can be defined by vectors of integers.
Alternatively a Petri net can be defined as a vector addition system with states (VASS).
A VASS is a directed graph in which the vertices are system states and the labeled edges are transition rules.
Starting with an initial vector of system parameters, the execution admitted by the VASS is a consecutive sequence of configuration transitions, where a configuration consists of a state and a vector of system parameters.
Many system properties can be studied using VASS model~\cite{Esparza1998}.
Decidability has been a central issue in these studies.
The most well-known result in this area is the decidability of the reachability problem.
In both theory and practice reachability is an important issue.
A variety of formal verification problems can be reduced to the VASS reachability problem~\cite{Schmitz2016c}.
In applications it is often necessary to know that unsafe configurations are never to be reached.

The decidability of the VASS reachability problem started to attract the attention of researchers in early 1970's~\cite{HopcroftPansiot1979,vanLeeuwen1974}.
A decidability proof was outlined, completed and improved by successive efforts of Sacerdote and Tenney~\cite{SacerdoteTenney1977}, Mayr~\cite{Mayr-1981}, Kosaraju~\cite{Kosaraju-1982} and Lambert~\cite{Lambert1992}.
The algorithm is known as KLMST decomposition.
Different decidability proofs have been studied by Leroux~\cite{Leroux2010,Leroux2011,Leroux2012}.

Complexity theoretical studies of the VASS reachability problem have also advanced significantly our understanding of the VASS model.
Here is a summary of the known upper bounds, some of which are quite new.
For the low dimensional VASS reachability problem, Haase, Kreutzer, Ouaknine and Worrell showed that the problem is NP-complete in the $1$-dimensional case~\cite{HaaseKreutzerOuaknineWorrell2009}.
Blondin, Finkel, G\"{o}ller, Haase and McKenzie proved that it is PSPACE complete in the $2$-dimensional case~\cite{BEFGHLMT2021}.
The problems are NL-complete if unary encoding is used~\cite{EnglertLazicTotzke2016,BEFGHLMT2021}.
For the high dimensional VASS reachability problem, Leroux and Schmitz proved that the $d$-dimensional VASS reachability problem is in $\mathbf{F}_{d+4}$~\cite{LerouxSchmitz2015,Schmitz2017,LerouxSchmitz2019} for all $d>2$.
As a consequence the general VASS problem, in which dimension is part of input, is in $\mathbf{F}_{\omega}$.
The lower bound of the high dimensional VASS reachability problem has also been investigated in~\cite{CzerwinskiLasotaLazicLerouxMazowiecki,czerwinski2021reachability,lasota2021improved,leroux2021reachability}.
It was proved by Czerwi{\'n}ski and Orlikowski that the general VASS reachability problem is Ackermann complete~\cite{czerwinski2021reachability}.
The $\mathbf{F}_{d+4}$ upper bound has been improved recently.
It is proved in~\cite{FuYangZheng2024} that the reachability problem of $d$-VASS is in $\mathbf{F}_d$ for all $d>2$.

In applications, the system parameters we are concerned with are not necessarily independent.
There may well be huge interdependency among the parameters.
Parameter dependency can be measured by geometric dimension.
If in a say $D$-dimensional VASS, the system parameters in every execution lie completely in a region that is essentially $d$-dimensional for some $d<D$, there are only $d$ independent parameters from which the other $D-d$ parameters can be derived.
We shall call such a VASS geometrically $d$-dimensional.
In this paper we are interested in the complexity upper bound of $g\mathbb{VASS}^d$ (the reachability problem of the geometrically $d$-dimensional VASSes).
The main contributions of the paper are stated in the following theorem, which strengthens the recent result proved in~\cite{FuYangZheng2024}.
\begin{theorem}\label{MAIN}
$g\mathbb{VASS}^d$ is in $\mathbf{F}_d$ for all $d>2$.
\end{theorem}
At the technical level, our contributions are threefold:
\begin{itemize}
    \item We provide a simpler proof that reachability in geometrically $2$-dimensional VASSes can be captured by linear path schemes of small size.
The new proof does not refer to the Farkas-Minkowski-Weyl Theorem~\cite{BEFGHLMT2021,FuYangZheng2024} and ought to be accessible to a wider audience.
    \item We provide a simpler analysis of the complexity of the KLMST algorithm.
Previous studies~\cite{FuYangZheng2024} refer to the Length Function Theorem~\cite{Schmitz2016b}.
    \item Our presentation of the KLMST algorithm sheds more light on the algebraic, combinatorial and geometric properties of VASS reachability, which are exploited in our new version of the KLMST algorithm.
\end{itemize}

Section~\ref{Sec-Preliminary} recaps some preliminaries about VASS reachability.
Section~\ref{Sec-Refinement-Operations} proposes an improvement of the KLMST decomposition method.
Section~\ref{Sec-Reachability-in-d-VASS} observes a geometric termination condition, based on which a new KLMST algorithm is designed.
Section~\ref{Sec-Geometrically-2-Dimensional-d-VASS}
offers a simpler proof that the geometrically $2$-dimensional VASSes admit short paths.
Section~\ref{Sec-Proliferation-Tree} presents a simple proof of the $\mathbf{F}_d$ upper bound of the KLMST decomposition approach.
Section~\ref{Sec-Conclusion} points out a few open problems.

\section{Preliminaries}\label{Sec-Preliminary}

Let $\mathbb{N}$ be the set of nonnegative integers, $\mathbb{Z}$ the set of integers, $\mathbb{Q}$ the set of rational numbers and $\mathbb{Q}_{\ge}$ the set of the nonnegative rational numbers.
Let $\mathbb{V}$ denote the set of variables for nonnegative integers.
We write the small letters $a,b,c,d,i,j,k,l,m,n$, the capital letters $A,B,C,D,L,M,N$, and $\imath,\jmath$ for elements in $\mathbb{N}$, $r,s,t$ for elements in $\mathbb{Z}$, $\ell$ for elements in $\mathbb{Q}$, and $x,y,z$ for elements in $\mathbb{V}$.
For $L\le M$ let $[L,M]$ be the set $\{L,\ldots,M\}$.
Let $[1,M]$ be abbreviated to $[M]$
and let $[M]_{0}$ be $\{0\}\cup[M]$.
The notation $|S|$ stands for the number of the elements of the finite set $S$.

We write $\mathbf{a},\mathbf{b},\mathbf{c},\mathbf{d}$ for $D$-dimensional vectors in $\mathbb{N}^D$, $\mathbf{r},\mathbf{s},\mathbf{t},\mathbf{u},\mathbf{v}$ for vectors in $\mathbb{Z}^D$, and $\mathbf{x},\mathbf{y},\mathbf{z}$ for vectors in $\mathbb{V}^D$.
All vectors are column vectors.
For $i\in[D]$ we write for example $\mathbf{a}(i)$ for the {\em $i$-th entry} of $\mathbf{a}$.
We often write $a_i$ for $\mathbf{a}(i)$.
So $\mathbf{a}=(a_1,a_2,\ldots,a_D)^{\dag}$, where $(\_)^{\dag}$ is the transposition operator.
Let $\mathbf{1}=(1,\ldots,1)^{\dag}$ and $\mathbf{0}=(0,\ldots,0)^{\dag}$.
If $\sigma$ is a finite string of vectors, $|\sigma|$ is the length of $\sigma$, and for $i\in[|\sigma|]$, $\sigma[i]$ is the $i$-th element appearing in $\sigma$, and $\sigma[i,j]$ is the sub-string $\sigma[i]\ldots\sigma[j]$.
The binary relation $\le$ on $\mathbb{Z}^D$, and on $\mathbb{N}^D$ as well, is the point-wise ordering.

Recall that the {\em $1$-norm} $\|\mathbf{r}\|_1$ of $\mathbf{r}$ is $\sum_{i\in[D]}|\mathbf{r}(i)|$.
The {\em $\infty$-norm} $\|\mathbf{r}\|_{\infty}$ is $\max_{i\in[D]}|\mathbf{r}(i)|$.
The {\em $1$-norm} $\|A\|_1$ of an integer matrix $A$ is $\sum_{i,j}|A(i,j)|$, and the {\em $\infty$-norm} $\|A\|_{\infty}$ is $\max_{i,j}|A(i,j)|$.

The vector space $\mathbb{Z}^D$ is divided into $2^D$ orthants.
For $\#_1,\#_2,\ldots,\#_D\in\{+,-\}$, the {\em orthant} $Z_{\#_1,\#_2,\ldots,\#_D}$ contains all the $D$-dimensional vectors $(r_1,r_2,\ldots,r_D)^{\dag}\in\mathbb{Z}^D$ such that, for each $i\in[D]$, $r_i\ge0$ if $\#_i=+$ and $r_i\le0$ if $\#_i=-$.
The {\em first orthant} refers to $Z_{+,+,\ldots,+}$.
For $\mathbf{v}\in\mathbb{Z}^D$, we write $Z_{\mathbf{v}}$ for the orthant $Z_{\#_1,\#_2,\ldots,\#_D}$ such that for each $i\in[D]$, the sign $\#_i$ is $+$ if $\mathbf{v}(i)\ge0$ and is $-$ otherwise.

In the following, we review the necessary background materials.

\subsection{Non-Elementary Complexity Class}
\label{s-Non-Elementary-Complexity-Class}

VASS reachability is not elementary~\cite{czerwinski2021reachability}.
To characterize the problem complexity theoretically, one needs complexity classes beyond the elementary class.
Schmitz introduced an ordinal indexed collection of complexity classes $\mathbf{F}_3,\mathbf{F}_4,\ldots,\mathbf{F}_{\omega},\ldots,\mathbf{F}_{\omega^2},\ldots,\mathbf{F}_{\omega^{\omega}},\ldots$ and showed that many problems arising in theoretical computer science are complete problems in this hierarchy~\cite{Schmitz2016a}.
In the above sequence $\mathbf{F}_{3}=\mathbf{Tower}$ and $\mathbf{F}_{\omega}=\mathbf{Ackermann}$.
The class $\mathbf{Tower}$ is closed under elementary reduction and $\mathbf{Ackermann}$ is closed under primitive recursive reduction.
For the purpose of this paper it suffices to say that $\mathbf{F}_d$, where $d>2$, contains all the problems whose space complexity is bounded by functions of the form
\[
f_{d-1}^{x+1}(x)=\underset{x+1}{\underbrace{f_{d-1}\ldots f_{d-1}}}(x),
\]
where $f_{d-1}\in\mathbf{F}_{d-1}$.
Define the tower function $\mathfrak{t}(x,y)$ as follows:
\begin{equation}\label{tower-2023-11-09}
\mathfrak{t}(x,0) = x,\text{ and }
\mathfrak{t}(x,y+1) = 2^{t(x,y)}.
\end{equation}
Now $F_{3}(x)\stackrel{\rm def}{=}\mathfrak{t}(x,x+1)\in\mathbf{F}_3$, and $F_{d}(x)\stackrel{\rm def}{=}F_{d-1}^{x+1}(x)\in\mathbf{F}_d$ for $d>3$.

For notational simplicity we shall hide constant factors when writing a complexity function.
For example we write $\texttt{poly}(n)$ for {\em some} polynomial of $n$, and write $2^{O(n)}$ for $2^{cn}$ for some constant $c$.
Accordingly we write for example $\texttt{poly}(n){\cdot}2^{O(n)}=2^{O(n)}$ to indicate that $\texttt{poly}(n){\cdot}2^{O(n)}$ is dominated by $2^{c'n}$ for some constant $c'$ that depends on the degree of $\texttt{poly}(n)$ and the constant factor in $2^{O(n)}$ of $\texttt{poly}(n){\cdot}2^{O(n)}$.

\subsection{Integer Programming}

We shall need a result in integer linear programming~\cite{Schrijver1986}.
Let $A$ be an $m\,{\times}\,k$ integer matrix and $\mathbf{x}\in\mathbb{V}^k$.
The {\em homogeneous equation system} of $A$ is given by the linear equation system $\mathcal{E}_0$ specified by
\begin{equation}\label{2019-05-15}
A\mathbf{x}=\mathbf{0}.
\end{equation}
A solution to~(\ref{2019-05-15}) is some nontrivial $\mathbf{n}\in\mathbb{N}^k\setminus\{\mathbf{0}\}$ such that $A\mathbf{n}=\mathbf{0}$.
Since the point-wise ordering $\le$ on $\mathbb{N}^k$ is a well founded well quasi order, $\mathcal{S}$ must be generated by a finite set of nontrivial minimal solutions.
This finite set is called the {\em Hilbert base} of $\mathcal{E}_0$, denoted by $\mathcal{H}(\mathcal{E}_0)$.
The following important result is proved by Pottier~\cite{Pottier1991}, in which $r$ is the {\em rank} of the matrix $A$.
\begin{lemma}[Pottier]\label{pottier-lemma}
$\|\mathbf{n}\|_1 \le (1\,{+}\,k{\cdot}\|A\|_{\infty})^r$ for every $\mathbf{n}\in\mathcal{H}(\mathcal{E}_0)$.
\end{lemma}

Let $\mathbf{r}\in\mathbb{Z}^k$.
Solutions to the {\em equation system} $\mathcal{E}$ defined by
\begin{equation}\label{hes}
A\mathbf{x}=\mathbf{r}
\end{equation}
can be derived from the Hilbert base of the homogeneous equation system $A\mathbf{x}-x'\mathbf{r}=\mathbf{0}$.
A solution to $\mathcal{E}$ is of the form
\[
\mathbf{m} + \sum_{\mathbf{n}\in\mathcal{H}(\mathcal{E}_0)}k_{\mathbf{n}}\mathbf{n},
\]
where $\mathbf{m}\in\mathcal{H}(\mathcal{E})$, the set of the minimal solutions to $\mathcal{E}$, and $k_{\mathbf{n}}\in\mathbb{N}$ for each $\mathbf{n}\in\mathcal{H}(\mathcal{E}_0)$.
The following is immediate from Lemma~\ref{pottier-lemma}.
\begin{corollary}\label{eq-sol}
$\|\mathbf{m}\|_1 \le (1+k{\cdot}\|A\|_{\infty}+\|\mathbf{r}\|_{\infty})^{r+1}$ for every $\mathbf{m}\in\mathcal{H}(\mathcal{E})$.
\end{corollary}

By Lemma~\ref{pottier-lemma} and Corollary~\ref{eq-sol}, if in~(\ref{2019-05-15}) and~(\ref{hes}) the number of variables is $2^{\texttt{poly}(n)}$ and the number of equations is $\texttt{poly}(n)$, then the minimal solutions are of $\texttt{poly}(n)$ size, and both $\mathcal{H}(\mathcal{E}_0)$ and $\mathcal{H}(\mathcal{E})$ are of $2^{\texttt{poly}(n)}$ size and can be guessed in $\texttt{poly}(n)$ space.

\subsection{VASS as Labeled Digraph}\label{Sec-Vector-Addition-System-with-States}

A {\em digraph} is a finite directed graph in which multi-edges and self loops are admitted.
A {\em $D$-dimensional vector addition system with states}, or {\em $D$-VASS}, is a labeled digraph $G=(Q,T)$ where $Q$ is the set of vertices and $T$ is the set of edges.
Edges are labeled by elements of $\mathbb{Z}^D$, and the labels are called {\em displacements}.
A {\em state} is identified to a vertex and a {\em transition} is identified to a labeled edge.
We write $o,p,q$ for states, and $e,f,g$ for edges.
A transition from $p$ to $q$ labeled $\mathbf{t}$ is denoted by $p\stackrel{\mathbf{t}}{\longrightarrow}q$.
The edge $p\stackrel{\mathbf{t}}{\longrightarrow}q$ is an {\em out-going edge} of $p$ and an {\em incoming edge} of $q$.
In the rest of the paper we refer to $G=(Q,T)$ either as a graph or as a VASS.
The input size of a VASS $G=(Q,T)$, denoted by $|G|$, is to be understood as the length of its binary code.
We say that a $D$-VASS is {\em strongly connected} if the graph is strongly connected.
We write $\|T\|$ for $\max_{p\stackrel{\mathbf{t}}{\longrightarrow}q\in T}\|\mathbf{t}\|_{\infty}$.
In other words, $\|T\|$ is the maximal absolute value of all the entry values of the transition labels of $G$.

A {\em path} in $G$ from $p$ to $q$ labeled $\pi=\mathbf{t}_1\ldots\mathbf{t}_k$ is an edge sequence
\begin{equation}\label{2023-12-09}
p=p_0\stackrel{\mathbf{t}_1}{\longrightarrow}p_1\stackrel{\mathbf{t}_2}{\longrightarrow}\ldots
\stackrel{\mathbf{t}_{k-1}}{\longrightarrow}p_{k-1}\stackrel{\mathbf{t}_{k}}{\longrightarrow}p_k=q
\end{equation}
for the unique states $p_1,\ldots,p_{k-1}$, often abbreviated to $p\stackrel{\pi}{\longrightarrow}q$.
We shall confuse $\pi$ to the path~(\ref{2023-12-09}).
The length of $\pi$ is denoted by $|\pi|$, which is the number of the edges appearing in $\pi$.
Let $\pi[i]$ denote $p_{i-1}\stackrel{\mathbf{t}_i}{\longrightarrow}p_i$, and let $\pi[i,j]$ denote $p_{i-1}\stackrel{\mathbf{t}_{i}}{\longrightarrow}p_{i}\stackrel{\mathbf{t}_{i+1}}{\longrightarrow}\ldots
\stackrel{\mathbf{t}_{j}}{\longrightarrow}p_{j}$.
We shall write $p\longrightarrow^{*}q$ for a path when the transition labels are immaterial.
A {\em circular path} (or {\em cycle}) $\O$ is a path $p\longrightarrow^{*}q$ such that $p=q$.
A {\em simple cycle} $\o$ is a cycle in which no vertex occurs more than once. 

A {\em Parikh image} for $G=(Q,T)$ is a function in $\mathbb{N}^{T}$.
We will write $\psi$ for Parikh images, and $\Psi\in\mathbb{V}^T$ for Parikh image variables.
The {\em displacement} $\Delta(\psi)$ is defined by $\sum_{e=(p,\mathbf{t},q)\in T}\psi(e){\cdot}\mathbf{t}$, and $\Delta(\Psi)$ by $\sum_{e=(p,\mathbf{t},q)\in T}\Psi(e){\cdot}\mathbf{t}$.
For a path $\pi$ we write $\Im(\pi)$ for the Parikh image defined by letting $\Im(\pi)(e)$ be the number of occurrences of $e$ in $\pi$.

\subsection{Constraint Graph Sequence}

Let $G=(Q,T)$ be a {\em strongly connected} $D$-VASS.
For $p,q\in Q$, not necessarily distinct, the triple $pGq$ is a {\em constraint graph (CG)}.
We call $p$ the {\em entry state} and $q$ the {\em exit state} of the CG.
When $G$ is the trivial graph with one state $p$ and no transition, $p=q$ and the CG is simply the state $p$.
A constraint graph is nothing but a strongly connected VASS with a specified entry state and a specified exit state.
Constraint graphs are the basic building bricks in the studies of VASS reachability.

In the constraint graph $pGq$ we are interested in paths from $p$ to $q$.
Path is a graph theoretical concept.
With VASS models one is interested not only in system states but also system parameters.
A {\em local state} $o(\mathbf{s})$ consists of a state $o$ and a {\em location} $\mathbf{s}\in \mathbb{Z}^D$.
If the initial state of~(\ref{2023-12-09}) is located at $\mathbf{s}_0$, the path $\pi$ can be annotated with the locations and becomes
\begin{equation}\label{2024-11-02}
p_0(\mathbf{s}_0)\stackrel{\mathbf{t}_1}{\longrightarrow}p_1(\mathbf{s}_1) \stackrel{\mathbf{t}_2}{\longrightarrow}\ldots \stackrel{\mathbf{t}_{k-1}}{\longrightarrow}p_{k-1}(\mathbf{s}_{k-1})\stackrel{\mathbf{t}_k}{\longrightarrow}p_k(\mathbf{s}_k),
\end{equation}
where $\mathbf{s}_i=\mathbf{s}_{i-1}+\mathbf{t}_{i}$ for all $i\in[k]$.
The {\em displacement} $\Delta(\pi)$ of $\pi$ is the vector $\sum_{i\in[k]}\mathbf{t}_i$ for the accumulated effect of the transitions of $\pi$.

\begin{figure*}[t]
\centering
\includegraphics[scale=0.7]{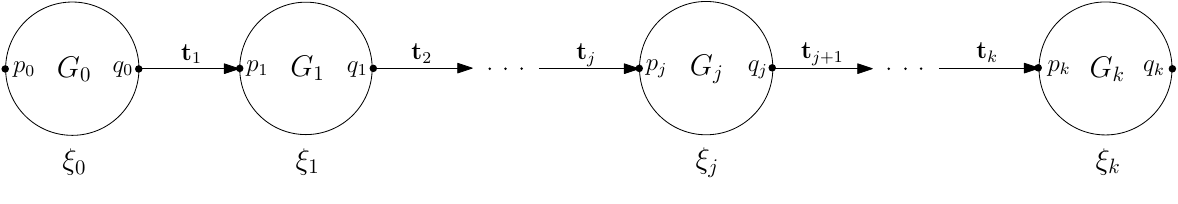}
\caption{Constraint Graph Sequence} \label{CGS}
\end{figure*}

A {\em constraint graph sequence (CGS) $\xi$} from $p=p_0$ to $q=q_k$ is a graph of the following form
\begin{equation}\label{CGS-2023-08-26}
\xi_0\stackrel{\mathbf{t}_1}{\longrightarrow}\xi_1\stackrel{\mathbf{t}_2}{\longrightarrow}\ldots
\stackrel{\mathbf{t}_{j-1}}{\longrightarrow}\xi_{j-1}\stackrel{\mathbf{t}_j}{\longrightarrow}\xi_{j}\stackrel{\mathbf{t}_{j+1}}{\longrightarrow}\ldots\stackrel{\mathbf{t}_k}{\longrightarrow}\xi_k,
\end{equation}
where $\xi_j=p_jG_jq_j$ and $G_j=(Q_j,T_j)$ is strongly connected for every $j\in[k]_{0}$, and $q_{j-1}\stackrel{\mathbf{t}_{j}}{\longrightarrow}p_{j}$ for every $j\in[k]$.
A diagrammatic illustration of~(\ref{CGS-2023-08-26}) is given in Figure~\ref{CGS}.
We call the edges $q_{0}\stackrel{\mathbf{t}_{1}}{\longrightarrow}p_{1}$, \ldots, $q_{k-1}\stackrel{\mathbf{t}_{k}}{\longrightarrow}p_{k}$ the {\em connecting edges}, and the CGes $\xi_0,\xi_1,\ldots,\xi_k$ the {\em component CGes} (or {\em components}) of $\xi$.
We will always assume that for each $j\in[k]_0$, attached to $\xi_j$ are a distinctive $D$-dimensional vector $\mathbf{x}_j$ of nonnegative integer variables for the entry location at $p_j$, a distinctive $D$-dimensional vector $\mathbf{y}_j$ of nonnegative integer variables for the exit location at $q_j$, and a distinctive injective function $\Psi_j$ from $T_j$ to $\mathbb{V}$ for the number of the occurrence of the edges.
We write $\mathbf{x}_j\xi_j\mathbf{y}_j$ or $p_j(\mathbf{x}_j)G_jq_j(\mathbf{y}_j)$ when we are concerned with the entry location and the exit location.
We write $\mathbf{a}\xi\mathbf{b}$ for the CGS obtained from $\xi$ by instantiating $\mathbf{x}_0$ and $\mathbf{y}_k$ by $\mathbf{a}$ and $\mathbf{b}$ respectively, and write $\mathbf{a}_j\xi_j\mathbf{b}_j$ or $p_j(\mathbf{a}_j)G_jq_j(\mathbf{b}_j)$ if $\mathbf{x}_j$ and $\mathbf{y}_j$ are instantiated by $\mathbf{a}_j$ and $\mathbf{b}_j$ respectively in the walk we are concerned with.
The size $|\xi|$ of $\xi$ is the size of the graph of $\xi$.
The size $|\mathbf{a}\xi\mathbf{b}|$ of $\mathbf{a}\xi\mathbf{b}$ is defined by $|\mathbf{a}|+|\xi|+|\mathbf{b}|$, where $|\mathbf{a}|$ for example is $\sum_{i\in[D]}|\mathbf{a}(i)|$.
We remark that a CGS differs from an KLM sequence~\cite{LerouxSchmitz2019} in that the components of the former must be strongly connected.

The CGS in~(\ref{CGS-2023-08-26}) imposes constraints on the shape of paths from $p_0$ to $q_k$.
A path $\pi$ is {\em admitted} by $\xi$ if $\pi=\pi_0\mathbf{t}_1\pi_1\mathbf{t}_2\pi_2\mathbf{t}_3\ldots\mathbf{t}_k\pi_k$ such that for each $j\in[k]$ the path $\pi_j$ is from $p_j$ to $q_j$ inside $G_j$.
In Section~\ref{Sec-Refinement-Operations} we will define operations on CGSes that adjoin new constraints on the old ones.
We shall be mostly interested in the special relationship between the CGSes defined next.
\begin{definition}
A CGS $\xi'$ {\em refines} a CGS $\xi$, notation $\xi'\sqsubseteq\xi$, if every path admitted by $\xi'$ is admitted by $\xi$.
\end{definition}

For $\mathbb{S}\subseteq\mathbb{Z}^D$, we write $p(\mathbf{r})\stackrel{\pi}{\longrightarrow}_{\mathbb{S}}q(\mathbf{s})$, or even $p(\mathbf{r})\longrightarrow_{\mathbb{S}}^*q(\mathbf{s})$, if $p(\mathbf{r})\stackrel{\pi}{\longrightarrow}q(\mathbf{s})$ passes only locations in $\mathbb{S}$.
By abusing terminology, we often say that $p(\mathbf{r})\stackrel{\pi}{\longrightarrow}q(\mathbf{s})$ is in $\mathbb{S}$ when $p(\mathbf{r})\stackrel{\pi}{\longrightarrow}_{\mathbb{S}}q(\mathbf{s})$.
For example, we may write $p_0(\mathbf{s}_0)\stackrel{\pi}{\longrightarrow}_{\mathbb{Z}^D}p_k(\mathbf{s}_k)$ or $p_0(\mathbf{s}_0)\longrightarrow_{\mathbb{Z}^D}^*p_k(\mathbf{s}_k)$ for the path in~(\ref{2024-11-02}).
If $\pi$ is a path admitted by $p'G'q'$ and $\xi$ refines $p'G'q'$, we will also say, when $p'G'q'$ is clear from context, that $\xi$ is a refinement of $\pi$.

\subsection{Configuration, Walk and Traversal}

If locations are system parameters and transitions are rules for system evolution, we should really be interested in paths in the first orthant. A {\em configuration} is a local state whose location is in $\mathbb{N}^D$.
A {\em walk} is a path in which all the local states in the walk are configurations.
A walk $\mathbf{w}$ is {\em proper} if no configuration appears more than once in $\mathbf{w}$.
Given $\mathbf{a},\mathbf{b}\in\mathbb{N}^D$, a {\em witness} to $\mathbf{a}\xi\mathbf{b}$ (cf. Figure~\ref{CGS}) is a walk of the form
\[
p(\mathbf{a})\longrightarrow_{\mathbb{N}^D}^{*}q_0(\mathbf{b}_0) \stackrel{\mathbf{t}_1}{\longrightarrow}p_1(\mathbf{a}_1)\longrightarrow_{\mathbb{N}^D}^{*}\ldots \stackrel{\mathbf{t}_k}{\longrightarrow} p_k(\mathbf{a}_k)\longrightarrow_{\mathbb{N}^D}^{*}q(\mathbf{b})
\]
admitted by $\xi$.
Let $\mathcal{W}_{\mathbf{a}\xi\mathbf{b}}$ denote the set of all witnesses to $\mathbf{a}\xi\mathbf{b}$, equivalently it is the set of walks from $p(\mathbf{a})$ to $q(\mathbf{b})$ admitted by $\xi$.
The problem $g\mathbb{VASS}^d$ can be defined in terms of $\mathcal{W}_{\mathbf{a}\xi\mathbf{b}}$ as follows:
\begin{quote}
Given a geometrically $d$-dimensional VASS $G$ and two configurations $p(\mathbf{a}),q(\mathbf{b})$, is
$\mathcal{W}_{p(\mathbf{a})Gq(\mathbf{b})}\ne\emptyset$\,?
\end{quote}
For $\mathbf{w}\in\mathcal{W}_{\mathbf{a}\xi\mathbf{b}}$, the notation $\mathbf{w}_j$ is for the sub-walk of $\mathbf{w}$ inside $\xi_j$.

A {\em traversal} in $\mathbf{a}\xi\mathbf{b}$ is a witness to $\mathbf{a}\xi\mathbf{b}$ that passes all the edges of $\xi$.
Quantitatively a traversal satisfies the {\em Euler Condition}, defined for each $j\in[k]_0$ by the following.
\begin{eqnarray}\label{2023-12-28}
\sum_{e=(p',\mathbf{t},q')\in T_j}\Psi_j(e)(\mathbf{1}_{q'}-\mathbf{1}_{p'}) &=& \mathbf{1}_{q_j}-\mathbf{1}_{p_j}, \label{lyjs-0} \\
\Psi(e) &>& 0, \ \text{ for every $e\in T_j$}, \label{all-edges} \\
\mathbf{x}_j+\Delta(\Psi_j) &=& \mathbf{y}_j. \label{lyjs-1}
\end{eqnarray}
where $\Psi_j$ is an injective function from $T_j$ to $\mathbb{V}$ and the notation $\mathbf{1}_{q_j}$ for example is a $|Q|$-dimensional indicator vector whose $q_j$-th entry is $1$ and all the other entries are $0$.
The variable $\Psi_j(e)$ is for the number of occurrence of the edge $e$.
The Euler Condition is necessary and sufficient for the existence of a path that enters $G_j$ in $p_j$ and leaves $G_j$ from $q_j$ and {\em passes all the edges of $\,G_j$}.
For each $j\in[k]$ there is an equation connecting the exit place of $\xi_{j-1}$ to the entry place of $\xi_{j}$:
\begin{eqnarray}
\mathbf{y}_{j-1}+\mathbf{t}_j &=& \mathbf{x}_j. \label{lyjs-2}
\end{eqnarray}
To specify that $p_0$ is located at $\mathbf{a}$ and $q_k$ is located at $\mathbf{b}$, introduce the equations
\begin{eqnarray}
\mathbf{x}_0 &=& \mathbf{a}, \label{lyjs-3} \\
\mathbf{y}_k &=& \mathbf{b}. \label{lyjs-4}
\end{eqnarray}
The {\em characteristic system} $\mathcal{E}_{\mathbf{a}\xi\mathbf{b}}$ for $\mathbf{a}\xi\mathbf{b}$ is defined by~(\ref{lyjs-0}), (\ref{all-edges}), (\ref{lyjs-1}), (\ref{lyjs-2}), (\ref{lyjs-3}) and~(\ref{lyjs-4}).
A solution to $\mathcal{E}_{\mathbf{a}\xi\mathbf{b}}$ is an assignment of {\em nonnegative} integers to the variables that renders valid~(\ref{lyjs-0}) through~(\ref{lyjs-4}).
The system $\mathcal{E}_{\mathbf{a}\xi\mathbf{b}}$ is an {\em under}-specification of the traversals in $\mathbf{a}\xi\mathbf{b}$ (cf. Figure~\ref{CGS}).
The reason that $\mathcal{E}_{\mathbf{a}\xi\mathbf{b}}$ does not completely characterize the traversals is that it does nothing to ensure that the paths inside each $\xi_j$ are in $\mathbb{N}^D$.
What are not captured equationally will be dealt with by other means, see Section~\ref{Sec-Coverability} and Section~\ref{Sec-Geometric-Dimension}.

The system $\mathcal{E}_{\mathbf{a}\xi\mathbf{b}}$ is {\em satisfiable} if it has a solution, it is {\em unsatisfiable} otherwise.
The reachability instance $\mathbf{a}\xi\mathbf{b}$ is (un)satisfiable if $\mathcal{E}_{\mathbf{a}\xi\mathbf{b}}$ is (un)satisfiable.
It should be emphasized that the unsatisfiability of $\mathcal{E}_{\mathbf{a}\xi\mathbf{b}}$ does not deny the existence of a witness to $\mathbf{a}\xi\mathbf{b}$, it only rules out the existence of any traversals in $\mathbf{a}\xi\mathbf{b}$.
The complexity of satisfiability testing is ignorable compared to that of reachability testing.
Our nondeterministic algorithm aborts whenever it produces an unsatisfiable CGS.
We will not mention the satisfiability testing issue any more.

Suppose $\mathbf{n}$ is a solution to $\mathcal{E}_{\mathbf{a}\xi\mathbf{b}}$.
We write $\mathbf{x}_j^{\mathbf{n}}$ for the vector of the nonnegative integers assigned to $\mathbf{x}_j$ by the solution $\mathbf{n}$.
The notations $\mathbf{y}_j^{\mathbf{n}},\Psi_j^{\mathbf{n}}$ are defined analogously.
More generally for a variable $z$ of the equation system, we write $z^{\mathbf{m}}$ for the value assigned to $z$ by $\mathbf{m}$.
The paths defined by $\Psi_j^{\mathbf{n}}$ are not unique, which will not be an issue for a nondeterministic algorithm.
Whenever we say ``a path defined by $\Psi_j^{\mathbf{n}}$'', we mean some path defined by $\Psi_j^{\mathbf{n}}$.
We shall confuse ``a solution to $\mathcal{E}_{\mathbf{a}\xi\mathbf{b}}$'' to ``a path from $\mathbf{a}$ to $\mathbf{b}$ admitted by $\xi$''.

A solution $\mathbf{m}'$ to $\mathcal{E}_{\mathbf{a}\xi\mathbf{b}}$ is {\em over} another solution $\mathbf{m}$ to $\mathcal{E}_{\mathbf{a}\xi\mathbf{b}}$, notation $\mathbf{m}\le\mathbf{m}'$, if $z^{\mathbf{m}}\le z^{\mathbf{m}'}$ for every variable of $\mathcal{E}_{\mathbf{a}\xi\mathbf{b}}$.
If $\mathbf{m}\ne\mathbf{m}'$, the difference $\mathbf{m}'-\mathbf{m}$ is a solution to the {\em homogeneous characteristic system} $\mathcal{E}_{\mathbf{a}\xi\mathbf{b}}^0$ obtained from $\mathcal{E}_{\mathbf{a}\xi\mathbf{b}}$ by changing the constant terms to $0$.
The system $\mathcal{E}_{\mathbf{a}\xi\mathbf{b}}^0$ contains the equations
\begin{eqnarray}
\mathbf{x}_0 &=& \mathbf{0}, \label{hcs-1} \\
\mathbf{y}_k &=& \mathbf{0}, \label{hcs-2}
\end{eqnarray}
and for all $j\in[k]_0$, the equations
\begin{eqnarray}
\sum_{e=(p',\mathbf{t},q')\in T_j}\Psi_j(e)(\mathbf{1}_{q'}-\mathbf{1}_{p'}) &=& \mathbf{0}, \label{hcs-3} \\
\mathbf{x}_j+\Delta(\Psi_j) &=& \mathbf{y}_j, \label{hcs-4} \\
\mathbf{y}_{j-1} &=& \mathbf{x}_j. \label{hcs-5}
\end{eqnarray}
A solution to $\mathcal{E}_{\mathbf{a}\xi\mathbf{b}}^0$ is a {\em nontrivial} assignment of nonnegative integers to the variables rendering valid~(\ref{hcs-1}) through~(\ref{hcs-5}).

For complexity theoretical studies, the size of the minimal solutions to $\mathcal{E}_{\mathbf{a}\xi\mathbf{b}}$/$\mathcal{E}_{\mathbf{a}\xi\mathbf{b}}^0$ is highly relevant.
\begin{lemma}\label{2024-01-08-sln}
$\|\mathbf{m}\|_1=2^{\texttt{poly}(|\mathbf{a}\xi\mathbf{b}|)}$ for every $\mathbf{m}$ in $\mathcal{H}(\mathcal{E}_{\mathbf{a}\xi\mathbf{b}})\cup\mathcal{H}(\mathcal{E}_{\mathbf{a}\xi\mathbf{b}}^0)$.
\end{lemma}
\begin{proof}
There are $O(|\mathbf{a}\xi\mathbf{b}|)$ equations in $\mathcal{E}_{\mathbf{a}\xi\mathbf{b}}$, and in $\mathcal{E}_{\mathbf{a}\xi\mathbf{b}}^0$ as well, containing $O(|\mathbf{a}\xi\mathbf{b}|)$ variables.
The rank of the characteristic system is $O(|\mathbf{a}\xi\mathbf{b}|)$.
The norms $\|T\|$, $\|\mathbf{a}\|$ and $\|\mathbf{b}\|$ are bounded by $2^{|\mathbf{a}\xi\mathbf{b}|}$.
It follows from Corollary~\ref{eq-sol} that $\|\mathbf{m}\|_1=2^{\texttt{poly}(|\mathbf{a}\xi\mathbf{b}|)}$.
\end{proof}

\subsection{Linear CGS}
\label{Sec-Linear-CGS}

A directed graph is {\em circular} if it is strongly connected, contains at least one edge,  and every vertex of the graph has one incoming edge and one out-going edge.
It is {\em linear} if it is connected and every vertex has at most one incoming edge and at most one out-going edge.
Suppose $G=(Q,T)$ is a $D$-dimensional VASS.
If $G$ is circular and $q$ is a vertex of $G$, then $qGq$ is a {\em linear constraint graph} (LCG).
A {\em linear constraint graph sequence} (LCGS) is a CGS whose components are all LCGes.
An LCGS can be denoted by $\pi_0\O_1\pi_1\O_2\pi_2\ldots\pi_{k-1}\O_k\pi_k$ where $\pi_0\pi_1\pi_2\ldots\pi_{k-1}\pi_k$ is a {\em linear graph} and $\O_1,\O_2,\ldots,\O_k$ are {\em circular graphs}.
The following is an LCGS from $p$ to $q$ with three circular graphs.
\begin{center}
\includegraphics[scale=0.7]{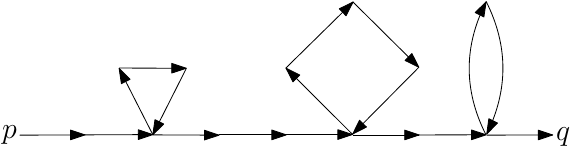}
\end{center}
An LCGS $\lambda=\pi_0\O_1\pi_1\O_2\pi_2\ldots\pi_{k-1}\O_k\pi_k$ is often defined by the edges of some VASS, where $\pi_0\pi_1\pi_2\ldots\pi_{k-1}\pi_k$ is a path admitted by the VASS, and $\O_1,\O_2,\ldots,\O_k$ are cycles admitted in the VASS.
In literature the linear CGSes are called linear path schemes~\cite{BEFGHLMT2021}.
We hope to see a linear path scheme as a graph because we will embed linear path schemes in CGSes. 

A traversal in an LCGS $\lambda=\pi_0\O_1\pi_1\O_2\pi_2\ldots\pi_{k-1}\O_k\pi_k$  is a walk admitted by $\lambda$ that passes all the edges of $\lambda$.
In other words, a traversal in an LCGS must pass every cycle at least once.
Unlike the general CGSes, the traversals in LCGSes can be completely characterized by linear systems.
To simplify the account, we introduce a map $\Delta^{\min}$ defined as follows:
For each path $\pi$ and each $i\in[D]$, let
\begin{eqnarray}\label{2024-01-04-minD}
\Delta^{\min}(\pi)(i) &=& \min_{j\in[|\pi|]}\left\{\Delta(\pi[1])(i)+\ldots+\Delta(\pi[j])(i)\right\}.
\end{eqnarray}
We need to check algebraically if there are traversals from $p(\mathbf{a})$ to $q(\mathbf{b})$ in $\mathbf{a}\lambda\mathbf{b}$.
The linear system for $\pi_0\O_1$ consists of three groups of inequations.
The first group is about $\pi_0$, it contains
    \begin{eqnarray*}
    \mathbf{a}+\Delta^{\min}(\pi_0) &\ge& \mathbf{0}.
    \end{eqnarray*}
The second group imposes the constraint that the first lap of $\O_1$ lies completely in the first orthant.
The inequation is
    \begin{eqnarray*}
    \mathbf{a}+\Delta(\pi_0) + \Delta^{\min}(\O_1) &\ge& \mathbf{0}.
    \end{eqnarray*}
The third group ensures that the last lap of $\O_1$ lies completely in the first orthant.
The inequations are
    \begin{eqnarray*}
    x_1&\ge& 1, \\
    \mathbf{a}+\Delta(\pi_0)+(x_1-1)\Delta(\O_1)  + \Delta^{\min}(\O_1) &\ge& \mathbf{0}.
    \end{eqnarray*}
Since both the first and the last laps defined by $\O_1$ are in the first orthant, the whole path $\pi_0\O_1^{x_1}$ lies in the first orthant.
Continuing in the above fashion we eventually get a linear system.
Let's call it the {\em witness system} of $\mathbf{a}\lambda\mathbf{b}$.
The system has $O(k)$ many variables.
By introducing no more than $O(k)$ variables, the inequations can be turned into equations.
By Corollary~\ref{eq-sol} one can find out the minimal solutions in $\texttt{poly}(|\mathbf{a}\lambda\mathbf{b}|)$ space.

The LCGSes play a crucial role in the studies of $\mathbb{VASS}^2$~\cite{BlondinFinkelGoellerHaaseMcKenzie2015}.
The key technical lemma of~\cite{BlondinFinkelGoellerHaaseMcKenzie2015} can be stated as follows.
\begin{lemma}\label{2023-12-27}
If there is a path from $o(\mathbf{r})$ to $o'(\mathbf{r}')$ admitted by a $2$-dimensional VASS $G=(Q,T)$, there is a path admitted by an LCGS $\lambda=\pi_0\O_1\pi_1\O_2\pi_2\ldots\pi_{k-1}\O_k\pi_k$ from $o(\mathbf{r})$ to $o'(\mathbf{r}')$ such that
\begin{itemize}
\item $\lambda$ refines $oGo'$,
\item $k=O(|Q|^2)$, and
\item $|\lambda|=(|Q|+\|T\|)^{O(1)}$.
\end{itemize}
Moreover if $o=o'$, then $k\le2$ and $\Delta(\O_1),\Delta(\O_2)\in Z_{\mathbf{r}'-\mathbf{r}}$.
\end{lemma}

The most remarkable thing about Lemma~\ref{2023-12-27} is that the size $(|Q|+\|T\|)^{O(1)}$ depends only on the VASS, it does {\em not} depend on $\mathbf{r},\mathbf{r}'$!
There could be infinitely many proper walks from $p(\mathbf{r})$ to $q(\mathbf{r}')$, and $\|\mathbf{r}\|,\|\mathbf{r}'\|$ could be extremely large.
The lemma points out however that there is a {\em finite} set $\mathcal{L}$ of LCGSes of size bounded by $(|Q|+\|T\|)^{O(1)}$ such that if there is a traversal from $p(\mathbf{r})$ to $q(\mathbf{r}')$, then there is a traversal from $p(\mathbf{r})$ to $q(\mathbf{r}')$ that is a witness to an LCGS in $\mathcal{L}$.
In the light of the fact that each element of $\mathcal{L}$ is a {\em linear refinement} of $oGo'$, Lemma~\ref{2023-12-27} can be interpreted as saying that every $2$-dimensional CG admits short paths.

\subsection{Fixed Entries and Bounded Edges}
\label{Sec-Fixed-Entries-and-Bounded-Edges}

Let $\mathbf{m}\in\mathcal{H}(\mathcal{E}_{\mathbf{a}\xi\mathbf{b}})$ be such that $\mathbf{m}\le\mathbf{w}$, where $\mathbf{w}$ is seen as a solution to $\mathcal{E}_{\mathbf{a}\xi\mathbf{b}}$.
Notice that $\mathbf{w}-\mathbf{m}$ is a solution to $\mathcal{E}_{\mathbf{a}\xi\mathbf{b}}^0$ if $\mathbf{w}\ne\mathbf{m}$.
If $\mathbf{w}-\mathbf{m}$ is a solution to $\mathcal{E}_{\mathbf{a}\xi\mathbf{b}}^0$, then $\mathbf{w}-\mathbf{m}$ must be over some $\mathbf{o}\in\mathcal{H}(\mathcal{E}_{\mathbf{a}\xi\mathbf{b}}^0)$.
By repeating the argument, one decomposes $\mathbf{w}$ as follows:
\begin{equation}\label{2024-01-31}
\mathbf{w} = \mathbf{m} + \sum_{\mathbf{o}\in\mathcal{H}(\mathcal{E}_{\mathbf{a}\xi\mathbf{b}}^0)}l_{\mathbf{o}}{\cdot}\mathbf{o},
\end{equation}
where $l_{\mathbf{o}}$ is a nonnegative integer for each $\mathbf{o}\in\mathcal{H}(\mathcal{E}_{\mathbf{a}\xi\mathbf{b}}^0)$.

The minimal solution sets $\mathcal{H}(\mathcal{E}_{\mathbf{a}\xi\mathbf{b}}),\mathcal{H}(\mathcal{E}_{\mathbf{a}\xi\mathbf{b}}^0)$ reveal a lot of properties about the traversals in $\mathbf{a}\xi\mathbf{b}$.
To define these properties, it is convenient to consider the solution
\begin{eqnarray}\label{2024-01-04}
\mathbf{O}_{\mathbf{a}\xi\mathbf{b}} &=& \sum_{\mathbf{o}\in\mathcal{H}(\mathcal{E}_{\mathbf{a}\xi\mathbf{b}}^0)}\mathbf{o}
\end{eqnarray}
that takes into account of all the minimal solutions to $\mathcal{E}_{\mathbf{a}\xi\mathbf{b}}^0$.
Since the minimal solutions are of polynomial size and there are an exponential number of them, the size of the solution $\mathbf{O}_{\mathbf{a}\xi\mathbf{b}}$ must be polynomial.
We would like to know how $\mathbf{w}$ differs from $\mathbf{m}$.
One difference is to do with the entry/exit locations of say $\xi_j$.
\begin{definition}
The entry $i\in[D]$ is {\em unbounded} at $p_j$, respectively $q_j$, if $\mathbf{x}_j^{\mathbf{O}_{\mathbf{a}\xi\mathbf{b}}}(i)>0$, respectively $\mathbf{y}_j^{\mathbf{O}_{\mathbf{a}\xi\mathbf{b}}}(i)>0$;
otherwise it is {\em bounded} at $p_j$, respectively $q_j$.
\end{definition}
It is clear from~(\ref{2024-01-31}) and (\ref{2024-01-04}) that $\mathbf{x}_j^{\mathbf{w}}(i)=\mathbf{x}_j^{\mathbf{m}}(i)$ whenever $i$ is bounded at $p_j$, and that $\mathbf{y}_j^{\mathbf{w}}(i)=\mathbf{y}_j^{\mathbf{m}}(i)$ whenever $i$ is bounded at $q_j$.
In more general terms, every witness to $\mathbf{a}\xi\mathbf{b}$ that is over $\mathbf{m}$ enters $p_j$ at a location whose $i$-th entry is the fixed value $\mathbf{x}_j^{\mathbf{m}}(i)$, and it leaves $q_j$ at a location whose $i$-th entry is the fixed value $\mathbf{y}_j^{\mathbf{m}}(i)$.
The solution $\mathbf{O}_{\mathbf{a}\xi\mathbf{b}}$ also reveals similar property for the edges.

\begin{definition}
An edge $e$ in $\xi_j$ is {\em unbounded} if \,$\Psi_j^{\mathbf{O}_{\mathbf{a}\xi\mathbf{b}}}(e)>0$; it is {\em bounded} otherwise.
The CG $\xi_j$ is {\em bounded} if it contains a bounded edge, it is {\em unbounded} otherwise.
The CGS $\xi$ is {\em bounded} if some $\xi_j$ is bounded, it is {\em unbounded} otherwise.
\end{definition}
Again it follows from~(\ref{2024-01-31}) and (\ref{2024-01-04}) that every witness to $\mathbf{a}\xi\mathbf{b}$ that is over $\mathbf{m}$ must pass a bounded edge $e$ for precisely $\Psi_j^{\mathbf{m}}(e)$ times.

\section{Refinement Operations}
\label{Sec-Refinement-Operations}

The central theme of our KLMST algorithm is to refine the input CGS $\xi$ repeatedly until it has constructed a refinement sequence
\begin{equation}\label{2024-10-28}
\xi_k\sqsubseteq\xi_{k-1}\sqsubseteq\ldots\sqsubseteq\xi_1\sqsubseteq\xi
\end{equation}
such that $\mathbf{a}\xi_k\mathbf{b}$ is either unsatisfactory or both satisfactory and good.
Intuitively $\mathbf{a}\xi_k\mathbf{b}$ is good if there is provably a witness to it.
If there is a witness to $\mathbf{a}\xi\mathbf{b}$, there is an execution of the nondeterministic algorithm that, upon receiving $\mathbf{a}\xi\mathbf{b}$, produces a CGS sequence like~(\ref{2024-10-28}) such that $\mathbf{a}\xi_k\mathbf{b}$ is satisfactory and good.
If there is no witness to $\mathbf{a}\xi\mathbf{b}$, every execution sequence produced by the algorithm ends up with a CGS say $\xi_k$ such that $\mathbf{a}\xi_k\mathbf{b}$ is unsatisfactory.
We shall present a {\em witness based refinement} approach that removes not only bad paths but also some good paths, so long as the status of reachability remains unchanged.
In the rest of the section, we fix the CGS $\xi$ as shown in Figure~\ref{CGS}, and fix a witness $\mathbf{w}$ to $\mathbf{a}\xi\mathbf{b}$ and a minimal solution $\mathbf{m}$ to $\mathcal{E}_{\mathbf{a}\xi\mathbf{b}}$ such that $\mathbf{m}\le\mathbf{w}$.
We shall use $\mathbf{w}$ and $\mathbf{m}$ to argue for the soundness of the refinement operations.

We will classify the refinements into two forms, simplifications and decompositions.
The former are local manipulations on graphs that do {\em not} increase graph size.
The latter are global operations on CGSes that may incur an {\em exponential blowup} in graph size.
Two simplification operations and two decomposition operations will be introduced.

\subsection{Eulerian Simplification}
\label{Sec-Eulerian-Simplification}

The witness $\mathbf{w}$ to $\mathbf{a}\xi\mathbf{b}$ may not pass all the edges in $\xi_j=p_jG_jq_j$.
In order to carry out the satisfiability testing, we must delete all the edges of $G_j$ that do not appear in $\mathbf{w}$.
After the deletion the new graph say $G_j^{-}$ might not be strongly connected.
One may think of $G_j^{-}$ as a DAG where each vertex is a {\em strongly connected component} (SCC).
There is a polynomial time algorithm that computes the class of the SCCes.
The SCCs are connected by the edges of $G$ that are not in any SCC.
Starting from the entry state $p_j$ there may be more than one way to go from an SCC to another and then stop at the exit state $q_j$.
Every linearization of the SCCes with the edges connecting them form a CGS $\xi_j'$ that admits $\mathbf{w}_j$.
Since $\xi_j'$ is a subgraph of $G_j^{-}$,  $|\xi_j'|\le|G_j^{-}|$.
The {\em Eulerian simplification} $\xi'$ of $\xi$ is obtained from $\xi$ by substituting $\xi_j'$ for $\xi_j$.
The next lemma is evident.
\begin{lemma}\label{Eulerian-Splf}
$\xi'\sqsubseteq\xi$, $|\xi'|\le|\xi|$, and $\mathbf{w}\in\mathcal{W}_{\mathbf{a}\xi'\mathbf{b}}$.
\end{lemma}

\subsection{Orthogonal Simplification}
\label{Sec-Eulerian-Simplification}

Suppose $\o$ is a simple cycle in a $D$-dimensional VASS $G$.
Let
\[\bot_{\o} = \left\{i\in[D] \mid \Delta(\o)(i)=0\right\}.\]
By definition $\bot_{\o}$ is the set of the axes orthogonal to the displacement $\Delta(\o)$ of the simple cycle $\o$.
There may well be many simple cycles in $G$, hence the next definition.
\begin{definition}\label{our-rigidity}
$i\in[D]$ is {\em orthogonal} to $G$ if $\,i\in\bot_{\o}$ for every simple cycle $\o$ in $G$.
\end{definition}
\noindent Let $\bot_G$ be the set of all the indices in $[D]$ orthogonal to $G$.
When $G$ is clear from context, we will simply say that $i$ is orthogonal.
The nice thing about an orthogonal index $i$ is that all values in the $i$-th entry can be safely ignored.
Let's elaborate on this point with $\xi_j$.
Suppose $i$ is orthogonal to $G_j$.
There are two situations.
\begin{enumerate}
\item $i$ is bounded at $p_j$.
If $p_j(\mathbf{x}_j^{\mathbf{m}})\longrightarrow_{\mathbb{Z}^D}^*o(\mathbf{r})$ and $p_j(\mathbf{x}_j^{\mathbf{m}})\longrightarrow_{\mathbb{Z}^D}^*o(\mathbf{r}')$, then $\mathbf{r}(i)=\mathbf{r}'(i)$ must be valid.
To see this one only has to concatenate the two paths with a path from $o$ to $p_j$, which is possible due to the strong connectedness of $G_j$.
As a consequence of this observation, $i$ must also be bounded at $q_j$.
Another consequence is that a unique integer can be attached to a state of $G_j$, for instance the integer $\mathbf{r}(i)$ is attached to $o$.
This property is the reason for an alternative terminology for orthogonality, the so-called {\em rigidity}~\cite{LerouxSchmitz2019}.
At the algorithmic level, orthogonality ensures that the states attached with negative integers can be deleted from $G_j$ without sacrificing reachability.
When a state is deleted, so are the transitions attached to it, producing a new graph $G_j'$.
We then carry out a Eulerian simplification to $G_j'$, giving rise to a new CGS $\xi_j'$.
The {\em orthogonal simplification} $\xi’$ is obtained from $\xi$ by substituting $\xi_j'$ for $\xi_j$.

\item $i$ is unbounded at $p_j$.
For every simple cycle $\o$, the value $\|\Delta(\o)\|_{\infty}$ is bounded by $|Q|{\cdot}\|T\|$.
If $\mathbf{x}_j^{\mathbf{m}}(i)\le\mathbf{x}_j^{\mathbf{w}}(i)\le|Q|{\cdot}\|T\|$, the orthogonal simplification is the same as defined in the previous case.
If $\mathbf{x}_j^{\mathbf{w}}(i)>|Q|{\cdot}\|T\|$, then we impose the inequality $\mathbf{x}_j(i)>|Q|{\cdot}\|T\|$ in all future characteristic systems.
In this case we also say that the CGS $\xi$ imposed with the additional inequality
is an {\em orthogonal simplification} of $\xi$.
\end{enumerate}
After orthogonal simplification values in the $i$-th entry of $\xi_j'$ never drop below zero.
It is in this sense that the orthogonal entries in $\xi_j'$ can be ignored.
We have again the following evident fact.
\begin{lemma}\label{WeqWp}
$\xi'\sqsubseteq\xi$, $|\xi'|\le|\xi|$ and $\mathbf{w}\in\mathcal{W}_{\mathbf{a}\xi'\mathbf{b}}$.
\end{lemma}

There is a well-known construction that turns an index to an orthogonal index.
The trade-off is an increase of the number of states.
Suppose $i$ is not orthogonal to $G_j$.
If there is a function $U(n)$ such that every configuration $p'(\mathbf{r}')$ the walk $\mathbf{w}_j$ passes by satisfies the inequality $\mathbf{r}'(i)\le U(|\mathbf{a}\xi\mathbf{b}|)$, then $\mathbf{w}_j$ is said to be {\em $U(|\mathbf{a}\xi\mathbf{b}|)$-bounded in the $i$-th entry}.
The locations of the walk $\mathbf{w}_j$ must fall in the region
\[
\mathbb{B}_i \;\stackrel{\rm def}{=}\; \underset{i-1\ \mathrm{times}}{\underbrace{\mathbb{N}\times\ldots\times\mathbb{N}}} \times[0,U(|\mathbf{a}\xi\mathbf{b}|)] \times \underset{d-i\ \mathrm{times}}{\underbrace{\mathbb{N}\times\ldots\times\mathbb{N}}}.
\]
We turn $\xi_j=p_j(\mathbf{x}_j)G_jq_j(\mathbf{y}_j)$ to $\xi_j'=(p,\mathbf{x}_j^{\mathbf{w}}(i))(\mathbf{x}_j)\,G_j^{-i}\,(q,\mathbf{y}_j^{\mathbf{w}}(i))(\mathbf{y}_j)$, where the new VASS $G_j^{-i}=\left(Q_j^{-i},T_j^{-i}\right)$ is defined by
\begin{itemize}
\item $Q_j^{-i} = \left\{(p,c) \mid p\in Q\ \mathrm{and}\ 0\le c\le U(|\mathbf{a}\xi\mathbf{b}|)\right\}$, and
\item $T_j^{-i} =  \left\{(p,c)\stackrel{\mathbf{t}}{\longrightarrow}(q,c\,{+\,}\mathbf{t}(i))
\ \left|    \begin{array}{l}
p\stackrel{\mathbf{t}}{\longrightarrow}q \in T,\\
0\le c\le U(|\mathbf{a}\xi\mathbf{b}|), \\
0\le c\,{+\,}\mathbf{t}(i)\le U(|\mathbf{a}\xi\mathbf{b}|)
\end{array}\right\}\right.$.
\end{itemize}
Let $\xi'$ be obtained from $\xi$ by substituting $\xi_j'$ for $\xi_j$.
The next two lemmas follow immediately from the construction.
\begin{lemma}\label{2024-09-04}
The index $i$ is orthogonal to $G_j^{-i}$.
\end{lemma}
\begin{lemma}\label{2024-09-15}
$\xi'\sqsubseteq\xi$, $|\xi'|=O(U(|\mathbf{a}\xi\mathbf{b}|))$ and  $\mathbf{w}\in\mathcal{W}_{\mathbf{a}\xi'\mathbf{b}}$.
\end{lemma}

It should be pointed out that $i$ may not be bounded at $p_j$ and/or $q_j$.
That means that an algorithm that carries out the construction may have to guess $\mathbf{x}_j^{\mathbf{w}}(i)$ and/or $\mathbf{y}_j^{\mathbf{w}}(i)$.

\subsection{Algebraic Decomposition}
\label{Sec-Algebraic-Aspect-of-VASS-Reachability}

Suppose $\xi_j$ is unbounded.
Let $unbd(\xi_j)$ be the set of the unbounded edges in $\xi_j$.
The edges in $unbd(\xi_j)$ form {\em connected components} (CCes) of $G_j$.
The CCes are interconnected by the bounded edges.
Now $G_j$ is strongly connected.
If there was an unbounded edge that is not in any SCC composed of unbounded edges, every cycle containing the unbounded edge must contain a bounded edge.
That would imply that the unbounded edge was bounded.
We conclude that the edges in $unbd(\xi_j)$ form a disjoint set of SCCes.
The walk $\mathbf{w}_j$ inside $G_j$ is generally of the shape
\begin{equation}\label{2019-09-28}
p_j^0 G_j^0 q_j^0\stackrel{\mathbf{t}_j^1}{\longrightarrow}p_j^1 G_j^1 q_j^1\stackrel{\mathbf{t}_j^2}{\longrightarrow}\ldots \stackrel{\mathbf{t}_j^{k_j}}{\longrightarrow} p_j^{k_j}G_j^{k_j}q_j^{k_j},
\end{equation}
where $G_j^0,\ldots,G_j^{k_j}$ are multiple occurrences of the SCCes, $p_j^0=p_j$, $q_j^{k_j}=q_j$, and every bounded edge $e$ in $G_j$ appears in $\stackrel{\mathbf{t}_j^1}{\longrightarrow}\ldots \stackrel{\mathbf{t}_j^{k_j}}{\longrightarrow}$ for precisely $\Psi_j^{\mathbf{m}}(e)=\Psi_j^{\mathbf{w}}(e)$ times.
We say that~(\ref{2019-09-28}) is {\em decomposed} from $\xi_j$ {\em algebraically} or that~(\ref{2019-09-28}) is an {\em algebraic decomposition} of $\xi_j$.
Let $\xi'$ be obtained from $\xi$ by substituting~(\ref{2019-09-28}) for $\xi_j$.
A diagrammatic illustration of the decomposition is given in Figure~\ref{DCP}.
\begin{figure}[t]
\centering
\includegraphics[scale=1.1]{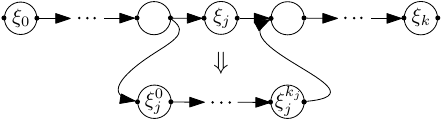}
\caption{Decomposition}  \label{DCP}
\end{figure}
The next lemma is assuring.
\begin{lemma}\label{my-decomposition-lemma}
$\xi'\sqsubseteq\xi$, $|\xi'|=2^{\texttt{poly}(|\mathbf{a}\xi\mathbf{b}|)}$ and  $\mathbf{w}\in\mathcal{W}_{\mathbf{a}\xi'\mathbf{b}}$.
\end{lemma}
\begin{proof}
The CGS in~(\ref{2019-09-28}) is more of a constraint than $\xi_j$ in the sense that a path is admitted by $\xi_j$ whenever it is admitted by the CGS in~(\ref{2019-09-28}).
The size bound is due to Lemma~\ref{2024-01-08-sln}.
\end{proof}

We have seen how refinement can be done using the information about the bounded edges.
We see next how refinement can be done using the boundedness information about the entry/exit locations.

\subsection{Combinatorial Decomposition}
\label{Sec-Coverability}

Assume that $p(\mathbf{a})Gq(\mathbf{\mathbf{b}})$ is strongly connected and unbounded, and for simplicity assume that no $i\in[D]$ is orthogonal to $G$.
If there are infinitely many proper witnesses to $p(\mathbf{a})Gq(\mathbf{\mathbf{b}})$, there are arbitrarily long proper witnesses to $p(\mathbf{a})Gq(\mathbf{\mathbf{b}})$.
Suppose $\mathbf{w}$ is as long as necessary to validate the following argument.
For each positive integer $N_1>\|\mathbf{a}\|_{\infty}$, by the pigeon hole principle, there must be an initial walk $p(\mathbf{a})\longrightarrow_{\mathbb{N}^D}^{*}q_1(\mathbf{a}_1)$ of $\mathbf{w}$ such that $\mathbf{a}_1(i_1)\ge N_1$ for some $i_1\in[D]$, and the length of $p(\mathbf{a})\longrightarrow_{\mathbb{N}^D}^{*}q_1(\mathbf{a}_1)$ is bounded by $|Q|{\cdot}N_1^D$.
Now ignore the $i_1$-th entry temporarily, and let $N_2$ be another large number.
Starting from $q_1(\mathbf{a}_1)$, there is a subwalk $q_1(\mathbf{a}_1)\longrightarrow_{\mathbb{N}^D}^*q_2(\mathbf{a}_2)$ such that $\mathbf{a}_2(i_2)\ge N_2$ for some $i_2\in[D]\setminus\{i_1\}$.
The length of $q_1(\mathbf{a}_1)\longrightarrow_{\mathbb{N}^D}^*q_2(\mathbf{a}_2)$ is bounded by $|Q|{\cdot}N_2^{D-1}$.
We can choose $N_1$ so large that $N_1-N_2$ is still very large.
Continue the argument this way, we derive that there is an initial walk $p(\mathbf{a})\longrightarrow_{\mathbb{N}^D}^*q_d(\mathbf{a}_d)$ of $\mathbf{w}$ such that $\mathbf{a}_d\ge\mathbf{a}+|Q|{\cdot}\|T\|{\cdot}\mathbf{1}$.
Since $G$ is strongly connected and unbounded, there is a walk $q_d(\mathbf{a}_d)\longrightarrow_{\mathbb{N}^D}^*p(\mathbf{c})$ whose length is less than $|Q|$.
Thus
\[
p(\mathbf{a})\longrightarrow_{\mathbb{N}^D}^*p(\mathbf{c}) \text{ and } \mathbf{c}\ge\mathbf{a}+\mathbf{1},
\]
which is a kind of pumpability property.
It is proved by Rackoff~\cite{Rackoff1978} that in order to check if $\mathbf{c}$ exists that validates the pumping property, one only has to check the walks of length bounded by a double exponential function of $|p(\mathbf{a})Gq(\mathbf{b})|$, decidable in exponential space.
If the dimension $D$ is fixed, Rackoff's checking algorithm can be done in polynomial space.

Having motivated the pumpability property, suppose now that the component $\xi_j=p_jG_jq_j$ is unbounded.
The situation is a little more complex than what is described in the above since there are three types of entries at $p_j$.
For the bounded entries, we can indeed discuss its pumpability.
For the unbounded entries, there is no way to do that, but we can lift the values in these entries to so large that they will never drop below $0$ throughout $\xi_j$.
If $i$ is bounded at $p_j$ and is orthogonal to $G_j$, we can ignore the values in the $i$-th entry.
We are led to focus on the sets of indices defined as follows:
\begin{eqnarray*}
\imath(p_j) &=& \{i \mid i\ \text{is bounded at $p_j$ but is not orthogonal to } G_j\}, \\
\imath(q_j) &=& \{i \mid i\ \text{is bounded at $q_j$ but is not orthogonal to } G_j\}.
\end{eqnarray*}
Since we will talk about pumpability for a subset $I\subseteq [D]$ of the entries, we need to generalize the notion of walk.
An {\em $I$-walk} is a path $p_0(\mathbf{r}_0)\stackrel{\mathbf{t}_1}{\longrightarrow}p_1(\mathbf{r}_1)\stackrel{\mathbf{t}_2}{\longrightarrow}
\ldots\stackrel{\mathbf{t}_g}{\longrightarrow}p_g(\mathbf{r}_g)$ such that for every $g'\in[g]_{0}$ and every $i\in I$ the inequality $\mathbf{r}_{g'}(i)\ge0$ is true.
We write $\mathbf{r}\ge_I\mathbf{r}'$ if $\mathbf{r}(i)\ge\mathbf{r}'(i)$ for all $i\in I$.

By definition, $\mathbf{x}_j^{\mathbf{w}}(i)=\mathbf{x}_j^{\mathbf{m}}(i)$ for each $i\in\iota(p_j)$, and $\mathbf{y}_j^{\mathbf{w}}(i)=\mathbf{y}_j^{\mathbf{m}}(i)$ for each $i\in\iota(q_j)$.
The equalities imply that we really only have to consider the pumpability with regards to the minimal solution $\mathbf{m}$.
\begin{definition}
The CG $\mathbf{x}_j\xi_j\mathbf{y}_j$ is {\em forward $\mathbf{m}$-pumpable} if an $\imath(p_j)$-walk $p_j(\mathbf{x}_j^{\mathbf{m}})\longrightarrow^{*}p_j(\mathbf{r})$ exists such that $\mathbf{r}\ge_{\imath(p_j)}\mathbf{x}_j^{\mathbf{m}}+\mathbf{1}$.
It is {\em backward $\mathbf{m}$-pumpable} if there is an $\imath(q_j)$-walk
$q_j(\mathbf{r}')\longrightarrow^{*}q_j(\mathbf{y}_j^{\mathbf{m}})$ such that $\mathbf{r}'\ge_{\imath(q_j)}\mathbf{y}_j^{\mathbf{m}}+\mathbf{1}$.
The CG $\mathbf{x}_j\xi_j\mathbf{y}_j$ is {\em $\mathbf{m}$-pumpable} if it is both forward $\mathbf{m}$-pumpable and backward $\mathbf{m}$-pumpable.
\end{definition}

Rackoff's algorithm was proposed for checking coverability~\cite{Rackoff1978}.
Leroux and Schmitz applied the algorithm to test pumpability~\cite{LerouxSchmitz2015}.
We state the key technical lemma on which the algorithm is based, and provide the proof for the sake of completeness.

\begin{lemma}\label{2019-05-09}
Let $p(\mathbf{r}_0)=p^0(\mathbf{r}_0)\longrightarrow p^1(\mathbf{r}_1)\longrightarrow \ldots\longrightarrow p^L(\mathbf{r}_L)$ be an $I$-walk admitted in $G$, where $G=(Q,T)$ is strongly connected.
Let $c=|I|$ and $A>2^{|G|+1}$.
\begin{itemize}
\item [(I)]
If for every $i\,{\in}\,I$, there is some $L'\,{\in}[L]_{0}$ such that $\mathbf{r}_{L'}(i)\ge A^{1+c^c}$, then there is some $L''\,{\in}[L]_{0}$ such that $\mathbf{r}_{L''}(i)> A-2^{|G|}$ for every $i\in I$ and $L''<A^{(c+1)^{(c+1)}}$.
\item [(II)]
Under the condition of (I), there is an $I$-walk $p(\mathbf{a})\stackrel{\sigma}{\longrightarrow}p(\mathbf{r})$ such that $\mathbf{r}(i)>A-2^{|G|+1}$ for every $i\in I$ and that $|\sigma|<A^{(c+1)^{(c+1)}}$.
\end{itemize}
\end{lemma}
\begin{proof}
(I) We consider only proper walks.
If $c=0$, there is nothing to prove; otherwise let $m\in[L-1]$ be the largest number rendering true the proposition $\forall m'{\in}[m].\forall i{\in}I.\mathbf{r}_{m'}(i)<A^{1+c^c}$.
By the pigeonhole principle the $I$-walk $p(\mathbf{r}_0)\longrightarrow^{*}p^{m}(\mathbf{r}_{m})$ is bounded by
\[|Q|{\cdot}\underset{c}{\underbrace{A^{1+c^c}\cdots A^{1+c^c}}} \,\le\, 2^{|G|}{\cdot}A^{(1+c^c)c} < A^{1+(1+c^c)c}\]
in length.
Let $c'=|I'|$, where $i=\left\{i\,{\in}\,I\mid \mathbf{r}_{m+1}(i)<A^{1+c^c}\right\}$.
Thus
\begin{equation}\label{2024-09-05}
\forall i\in I\setminus I'.\mathbf{r}_{m+1}(i)\ge A^{1+c^c}
\end{equation}
Now $c'<c$ by the maximality of $m$.
By induction we get for some $L''\in[L]$ the $I'$-walk $p^{m+1}(\mathbf{r}_{m+1})\longrightarrow^{*} p^{L''}(\mathbf{r}_{L''})$ whose length is bounded by $A^{(c'+1)^{(c'+1)}}$ and that $\mathbf{r}_{L''}(i)>A-2^{|G|}$ for all $i\in I'$.
Thus $p(\mathbf{r}_0)\longrightarrow^{*}p^{m+1}(\mathbf{r}_{m+1})\longrightarrow^{*}p^{L''}(\mathbf{r}_{L''})$ is bounded in length by
\begin{equation}\label{2023-08-22}
A^{1+(1+c^c)c} + A^{(c'+1)^{(c'+1)}} \le A^{1+(1+c^c)c} + A^{c^c}<A^{(c+1)^{(c+1)}}.
\end{equation}
And for all $i\in I\setminus I'$,
\[
\mathbf{r}_{L''}(i) \ >\ \mathbf{r}_{m+1}(i)-\|T\| {\cdot}A^{(c'+1)^{(c'+1)}}
 \ \stackrel{(\ref{2024-09-05})}{>}\ (A-2^{|G|})A^{c^c}
 \ >\ A-2^{|G|}.
\]
Therefore $p(\mathbf{r}_0)\longrightarrow^{*}p^{m+1}(\mathbf{r}_{m+1})\longrightarrow^{*}p^{L''}(\mathbf{r}_{L''})$ is an $I$-walk.

(II) Since $G$ is strongly connected, a walk $p^{L''}(\mathbf{r}_{L''})\stackrel{\pi}{\longrightarrow}p(\mathbf{r})$ exists rendering true the inequality $|\pi|<|Q|$.
Moreover for all $i\in I$,
\[
\mathbf{r}(i) \ >\ \mathbf{r}_{L''}(i)-\sum_{p'\stackrel{\mathbf{t}}{\longrightarrow}q'\in T}\|\mathbf{t}\|_{\infty} \
>\ A-2^{|G|}-2^{|G|} \
=\ A-2^{|G|+1}.
\]
It is easily seen that~(\ref{2023-08-22}) can be tightened to
\[
A^{1+(1+c^c)c} + A^{(c'+1)^{(c'+1)}}+|Q| < A^{1+(1+c^c)c} + A^{c^c}+A <A^{(c+1)^{(c+1)}}.
\]
We conclude that $p(\mathbf{r}_0)\longrightarrow^{*}p^{m+1}(\mathbf{r}_{m+1})\longrightarrow^{*}p^{L''}(\mathbf{r}_{L''})\longrightarrow^{*}p(\mathbf{r})$
is an $I$-walk whose length is bounded by $A^{(c+1)^{(c+1)}}$.
\end{proof}

By our assumption,  $\mathbf{m}\le\mathbf{w}\in\mathcal{W}_{\mathbf{a}\xi\mathbf{b}}$. If $\mathbf{x}_j\xi_j\mathbf{y}_j$ is forward $\mathbf{m}$-pumpable, set $A=\|\mathbf{x}_j^{\mathbf{m}}\|_{\infty}+2^{|G_j|+1}+1$ and set $B=A^{1+D^D}=2^{\texttt{poly}(|\mathbf{a}\xi\mathbf{b}|)}$, where the second equality is due to Lemma~\ref{2024-01-08-sln}.
According to the (II) of Lemma~\ref{2019-05-09}, there is a circular walk $p_j(\mathbf{x}_j^{\mathbf{m}})\longrightarrow^{*}p_j(\mathbf{c})$ such that $\mathbf{c}(i)\ge\mathbf{x}_j^{\mathbf{m}}(i)+1$ for all $i\in\iota(p_j)$.
If $\mathbf{x}_j\xi_j\mathbf{y}_j$ is not forward $\mathbf{m}$-pumpable, then by Lemma~\ref{2019-05-09}, there is some $i\in\imath(p_j)$ such that the values of the $i$-th entry of the locations throughout $\mathbf{w}_j$ are all in $[0,B]$.
For the same reason, if $\mathbf{x}_j\xi_j\mathbf{y}_j$ is not backward $\mathbf{m}$-pumpable, then there is some $i\in\imath(q_j)$ such that the values of the $i$-th entry of the locations throughout $\mathbf{w}_j$ are in $[0,B]$.
So if $\mathbf{x}_j\xi_j\mathbf{y}_j$ is not $\mathbf{m}$-pumpable, by using the construction defined in the last paragraph of Section~\ref{Sec-Eulerian-Simplification}, we can turn $\xi_j=p_jG_jq_j$ to some $\xi_j^{-i}$ to which $i$ is orthogonal by Lemma~\ref{2024-09-04}.
The size of $\xi_j^{-i}$ is bounded by $B=2^{\texttt{poly}(|\mathbf{a}\xi\mathbf{b}|)}$.
By applying the orthogonal simplification to $\xi_j^{-i}$, we can turn it to a CGS $\xi_j'$ of the form~(\ref{2019-09-28}) such that $|\xi_j'|\le|\xi_j^{-i}|$.
The diagram illustration is still the one in Figure~\ref{DCP}.
Let $\xi'$ be obtained from $\xi$ by substituting $\xi_j'$ for $\xi_j$.
We say that $\xi'$ is {\em decomposed} from $\xi$ {\em combinatorially}, or that $\xi'$ is a {\em combinatorial decomposition} of $\xi$.
The next lemma follows from Lemma~\ref{2024-09-15}.
\begin{lemma}\label{my-reduction-lemma}
$\xi'\sqsubseteq\xi$, $|\xi'|=2^{\texttt{poly}(|\mathbf{a}\xi\mathbf{b}|)}$ and  $\mathbf{w}\in\mathcal{W}_{\mathbf{a}\xi'\mathbf{b}}$.
\end{lemma}

\subsection{Normal CGS}
\label{Sec-Normal-CGS}

Do the four refinement operations suffice for an algorithm that checks VASS reachability?
The answer is positive.
If no refinement can be applied to a CGS, it is a good CGS.

A good CGS $\mathbf{a}\Xi\mathbf{b}$ must be satisfiable, meaning that it admits at least one path; moreover every path it admits can be expanded to a walk.
For a good CGS, satisfiability guarantees the existence of a witness.
Recall that in Section~\ref{Sec-Linear-CGS} it is shown that every solution to a witness system of an LCGS defines a witness to the LCGS.
This is taken into account in the following definition.
\begin{definition}
The reachability instance $\mathbf{a}\Xi\mathbf{b}$ is {\em normal} if it is satisfiable and there is some $\mathbf{n}\in\mathcal{H}(\mathcal{E}_{\mathbf{a}\Xi\mathbf{b}})$ such that for all $j\in[k]_0$, either $\Xi_j$ is an LCGS or $\Xi_j$ is unbounded and $\mathbf{x}_j\Xi_j\mathbf{y}_j$ is $\mathbf{n}$-pumpable.
\end{definition}
A standard result in the theory of VASS reachability is that normal reachability instance has short witnesses.
Even though our definition of normality is slightly different, the proof of this standard result remains unchanged.
\begin{figure}[t]
\centering
\includegraphics[scale=0.8]{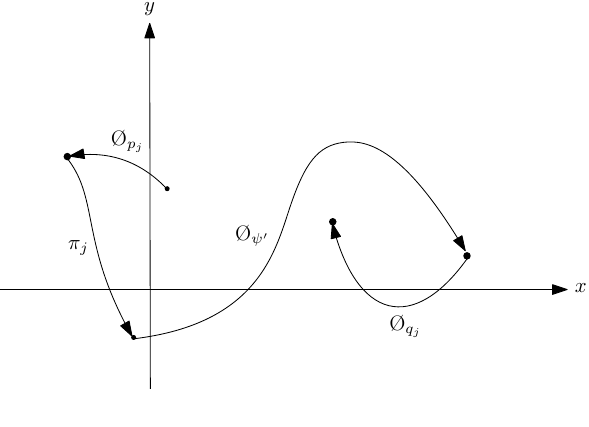}
\caption{Constructing Witness for Normal CGS.}  \label{cw4ncgs}
\end{figure}
\begin{proposition}\label{2024-11-01}
If $\mathbf{a}\Xi\mathbf{b}$ is a normal CGS, there is a witness to $\mathbf{a}\Xi\mathbf{b}$ whose length is bounded by $2^{\texttt{poly}(|\mathbf{a}\Xi\mathbf{b}|)}$.
\end{proposition}
\begin{proof}
Suppose $\mathbf{n}\in\mathcal{H}(\mathcal{E}_{\mathbf{a}\Xi\mathbf{b}})$ is such that every component of $\mathbf{a}\Xi\mathbf{b}$ is either an LCGS or unbounded and $\mathbf{n}$-pumpable.
The witness system of an LCG guarantees that every solution to $\mathcal{E}_{\mathbf{a}\Xi\mathbf{b}}$ defines a path whose restriction to an LCG stays in the first orthant.
There is nothing more to say about the LCGes in $\Xi$.
Now suppose $\Xi_j$ is unbounded and $\mathbf{x}_j\Xi_j\mathbf{y}_j$ is $\mathbf{n}$-pumpable for some $j\in[k]$.
By Lemma~\ref{2019-05-09} there is an $\iota(p_j)$-walk $p_j\stackrel{\O_{p_j}}{\longrightarrow}p_j$ in $\mathbf{x}_j\Xi_j\mathbf{y}_j$ such that $\Delta(\O_{p_j})\ge_{\iota(p_j)}\mathbf{1}$.
Symmetrically there is an $\iota(q_j)$-walk $q_j\stackrel{\O_{q_j}}{\longrightarrow}q_j$ in $\mathbf{x}_j\Xi_j\mathbf{y}_j$ such that $\Delta(\O_{q_j})\le_{\iota(q_j)}-\mathbf{1}$.
Let $\mathbf{O}_{\mathbf{a}\Xi\mathbf{b}}$ be defined as in~(\ref{2024-01-04}), and let $C$ be the minimal number such that for every $e\in T_j$,
\begin{equation}\label{2024-11-02-wan}
\left(C\Psi_j^{\mathbf{O}_{\mathbf{a}\Xi\mathbf{b}}}\right)(e) > \Im(\O_{p_j})(e)+\Im(\O_{q_j})(e).
\end{equation}
Let \[\psi'=C\Psi_j^{\mathbf{O}_{\mathbf{a}\Xi\mathbf{b}}} - \Im(\O_{p_j})-\Im(\O_{q_j}).\]
Since $\Psi_j^{\mathbf{O}_{\mathbf{a}\Xi\mathbf{b}}}$, $\Im(\O_{p_j})$ and $\Im(\O_{q_j})$ meet the equality~(\ref{lyjs-0}), $\psi'$ must meet the equality~(\ref{lyjs-0}).
Additionally the strict inequality~(\ref{2024-11-02-wan}) implies that $\psi'$ also meets the equality~(\ref{all-edges}).
We conclude that $\psi'$ defines a single cycle. 
Let the circular paths $p_j\stackrel{\O_{p_j}}{\longrightarrow}p_j$, $q_j \stackrel{\O_{q_j}}{\longrightarrow}q_j$ and $q_j \stackrel{\O_{\psi'}}{\longrightarrow}q_j$ be defined respectively by  $\Im(\O_{p_j})$, $\Im(\O_{q_j})$ and $\psi'$.
Let $p_j\stackrel{\pi_j}{\longrightarrow}q_j$ be a path defined by $\Psi_j^{\mathbf{n}}$.
Consider
\begin{equation}\label{2023-12-31-wan}
p_j\stackrel{\O_{p_j}}{\longrightarrow}p_j \stackrel{\pi_j}{\longrightarrow}q_j \stackrel{\O_{\psi'}}{\longrightarrow}q_j \stackrel{\O_{q_j}}{\longrightarrow}q_j.
\end{equation}
There is no guarantee that any of the paths $\O_{p_j},\pi_j,\O_{\psi'},\O_{q_j}$ in~(\ref{2023-12-31-wan}) is a walk.
In Figure~\ref{cw4ncgs}, a possible shape of~(\ref{2023-12-31-wan}) is projected at the $x$-axis and the $y$-axis (values in other axes are ignored), where $x$ is unbounded at $p_j$ and $y$ is unbounded at $q_j$.
Let $M=\|\Delta^{\min}(\pi_j)\|$ and $N=\|\Delta^{\min}(\O_{\psi'})\|$, confer~(\ref{2024-01-04-minD}) for the definition of $\Delta^{\min}(\_)$, and let $L=M+N$.
Consider the path
\[(\O_{p_j})^{2L}\pi_j(\O_{\psi'})^{2L}(\O_{q_j})^{2L}.\]
For each $i\in\iota(p_j)$, by the forward $\mathbf{n}$-pumpability,
\begin{eqnarray*}
\Delta\left((\O_{p_j})^{2L}\right)(i) &\ge& 2L, \end{eqnarray*}
and consequently
$\Delta\left((\O_{p_j})^{2L}\pi_j\O_{\psi'}\right)(i) \ge L$.
For each $i\in\iota(q_j)$, by the backward $\mathbf{n}$-pumpability,
\begin{eqnarray*}
\Delta\left((\O_{p_j})^{2L}\pi_j(\O_{\psi'})^{2L}\right)(i) &\ge& 2L,
\end{eqnarray*}
and consequently
$\Delta\left((\O_{p_j})^{2L}\pi_j(\O_{\psi'})^{k}\right)(i) \ge 2L$ for all $k\in[2L]$.
It follows that~(\ref{2023-12-31-wan}) is an $\left(\iota(p_j)\cap\iota(q_j)\right)$-walk.
If $p_j$ is unbounded at $i\in[D]$ that is not orthogonal, $\O_{p_j}$ might go into negative in the $i$-th entry when starting from $\mathbf{x}_j^{\mathbf{n}}$.
We need to uplift the entry location at $p_j$ high enough so that $(\O_{p_j})^{2L}$ is a walk.
Let $A_1=\|\Delta^{\min}((\O_{p_j})^{2L})\|$.
Symmetrically let $A_2=\|\Delta^{\min}((\O_{q_j})^{2L})\|$.
Let $\O_{\Psi}$ be a cycle defined by $\Psi_j^{\mathbf{O}_{\mathbf{a}\Xi\mathbf{b}}}$, and let $A_3=\|\Delta^{\min}(\O_{\Psi})\|$.
Set $A=A_1+A_2+A_3$.
Then
\begin{equation}\label{2024-01-05-noon}
p_j\stackrel{\left(\O_{\Psi}\right)^{A}}{\longrightarrow}p_j \stackrel{(\O_{p_j})^{2L}}{\longrightarrow}p_j \stackrel{\pi_j}{\longrightarrow}q_j \stackrel{(\O_{\psi'})^{2L}}{\longrightarrow}q_j \stackrel{(\O_{q_j})^{2L}}{\longrightarrow}q_j
\end{equation}
is a walk.

Let's take a look at the minimal size of the walk in~(\ref{2024-01-05-noon}).
By Lemma~\ref{2019-05-09}, the cycles $\O_{p_j}$ and $\O_{q_j}$ are bounded by $2^{\texttt{poly}(|\mathbf{a}\Xi\mathbf{b}|)}$ in length.
By Lemma~\ref{eq-sol}, the length of the cycle  $\O_{\psi'}$, the length of the path $\pi_j$, and $C$ as well are bounded by $2^{\texttt{poly}(|\mathbf{a}\Xi\mathbf{b}|)}$.
Consequently the values $L,M,N$ are bounded by $2^{\texttt{poly}(|\mathbf{a}\Xi\mathbf{b}|)}$.
For the same reason $A_1,A_2,A_3$ and $A$ as well are bounded by $2^{\texttt{poly}(|\mathbf{a}\Xi\mathbf{b}|)}$.
We conclude that there is a witness to the normal CGS $\mathbf{a}\Xi\mathbf{b}$ whose length is bounded by $2^{\texttt{poly}(|\mathbf{a}\Xi\mathbf{b}|)}$.
\end{proof}

\section{Geometrically $d$-Dimensional $\mathbb{VASS}$ is in $\mathbf{F}_d$}
\label{Sec-Reachability-in-d-VASS}

It is easily seen that the refinement sequence~(\ref{2024-10-28}) must be finite:
simplifications reduce graph size; algebraic decompositions generate smaller CGes; and combinatorial decompositions turn some indices to orthogonal.
A metric can be defined in terms of these parameters, which suffices for a finiteness argument for~(\ref{2024-10-28}).
If one is concerned with the complexity, one should really be interested in how a tight bound on the length of the sequence in~(\ref{2024-10-28}) can be derived and when the refinement sequence ought to be stopped.
This will be discussed respectively in Section~\ref{Sec-Geometric-Dimension} and in Section~\ref{sec-Leaf}.
Using the results from Section~\ref{Sec-Geometric-Dimension} and Section~\ref{sec-Leaf}, an $\mathbf{F}_d$ upper bound on the minimal size of the good CGSes will be derived in Section~\ref{sec-SuanFa}.
This upper bound suggests an extremely simple algorithm for $g\mathbb{VASS}^d$.

\subsection{Geometry of Decomposition}
\label{Sec-Geometric-Dimension}

Let $G=(Q,T)$ be a $D$-VASS.
A path in $G$ contains a sequence of simple cycles in $G$ and some edges that are not in any cycles.
The number of the latter is less than $|Q|$.
Let $V_{G}$ be the subspace of $\mathbb{Q}^D$ spanned by the displacements of the simple cycles in $G$.
By definition if $i$ is orthogonal to $G_j$, the $i$-th axis is orthogonal to the space $V_{G_j}$.
Let $d_{G}$ be the dimension of $V_{G}$.
We say that $G$ is \emph{geometrically $d_{G}$-dimensional}.
The following fact was observed by Leroux and Schmitz~\cite{LerouxSchmitz2019}.
We present an outline of the proof.
\begin{lemma}\label{decomp-2023-05-11}
Let $G_j'$ be the subgraph of $G_j$ defined by the unbounded edges of $\,G_j$.
If $G_j'\ne G_{j}$, then $d_{G_j'}<d_{G_j}$.
\end{lemma}
\begin{proof}
We shall refer to the homogeneous solution $\mathbf{O}_{\mathbf{a}\xi\mathbf{b}}$ defined in~(\ref{2024-01-04}).
Assume that $d_{G_j'}=d_{G_j}$.
Then $V_{G_j'}=V_{G_j}$.
Let $\O_j$ be a cycle containing every edge of the strongly connected graph $G_j$.
By the assumption, $\Delta(\O_j)$ is a linear combination of $\{\Delta(\o)\}_{\o\in\mathbf{C}}$, where $\mathbf{C}$ is the set of the simple cycles of $G_j'$.
Equivalently,
\begin{equation}\label{second-last}
k\Delta(\O_j)=\sum_{\o\in\mathbf{C}}r_{\o}{\cdot}\Delta(\o),
\end{equation}
where the coefficients are integers and $k$ is positive.
Let $k'$ be the minimal positive integer such that $\left(\Psi_j^{k'\mathbf{O}_{\mathbf{a}\xi\mathbf{b}}}-\sum_{\o\in\mathbf{C}}r_{\o}\Im(\o)\right)(e)>0$ for every edge $e$ in $G_j'$.
It follows from~(\ref{second-last}) that
\[\Delta\left(\Psi_j^{k'\mathbf{O}_{\mathbf{a}\xi\mathbf{b}}}\right) = k\Delta(\O_j)+\left(\Delta\left(\Psi_j^{k'\mathbf{O}_{\mathbf{a}\xi\mathbf{b}}}\right)-\sum_{\o\in\mathbf{C}}r_{\o}\Delta(\o)\right).\]
So from the homogeneous solution $k'\mathbf{O}_{\mathbf{a}\xi\mathbf{b}}$ one defines another homogeneous solution showing that every edge of $G_j$ is unbounded.
The contradiction refutes the assumption $d_{G_j'}=d_{G_j}$.
\end{proof}

It follows from Lemma~\ref{decomp-2023-05-11} that if a node is obtained from its parent by an algebraic decomposition, its geometrical dimension is strictly smaller than that of its parent.
We are now establishing the same property for the combinatorial decompositions.
\begin{lemma}\label{2024-01-17}
If $p_{j'}G_{j'}q_{j'}$ is obtained from $p_{j}G_{j}q_{j}$ by a combinatorial decomposition, then $d_{G_j'}<d_{G_j}$.
\end{lemma}
\begin{proof}
Let $G_j$ be a geometrically $c$-dimensional $D$-VASS.
If $c=D$, the result follows by the definition of combinatorial decomposition since $G_{j'}$ is essentially $(D{-}1)$-dimensional.
Assume that $c<D$.
By the definition of combinatorial decomposition some $k$ that is not orthogonal to $G_j$ is $B$-bounded in the $k$-th entry, where $B=2^{O(|\xi|)}$.
The bound incurs an orthogonal simplification.
Let $\mathbf{M}$ be the $(D{\times}c)$-matrix whose column vectors $\mathbf{c}_1,\ldots,\mathbf{c}_c$ are linearly independent displacements of $c$ simple cycles in $G_j$.
Let $\mathbf{r}_1,\ldots,\mathbf{r}_D$ be the row vectors of $\mathbf{M}$.
There must be $c$ linearly independent row vectors.
Without loss of generality, let them be $\mathbf{r}_1,\ldots,\mathbf{r}_c$. 
For each $j\in[D]\setminus[c]$, $\mathbf{r}_j$ is a linear combination of $\mathbf{r}_1,\ldots,\mathbf{r}_c$.
Therefore
$$
    \mathbf{M}=\begin{bmatrix}
        \mathbf{c}_1 \cdots \mathbf{c}_c
    \end{bmatrix}=\begin{bmatrix}
                \mathbf{r}_1 \\
        \vdots\\
        \mathbf{r}_D
    \end{bmatrix}= \begin{bmatrix}
        \mathbf{I}_c\\ \mathbf{R}_0
    \end{bmatrix}\cdot \begin{bmatrix}
        \mathbf{r}_1 \\
        \vdots\\
        \mathbf{r}_c
    \end{bmatrix}
$$
for some $(D-c)\times c$ matrix $\mathbf{R}_0$,
where $\mathbf{I}_c$ is the $c\times c$ identity matrix.

Suppose that the values in the $k$-th entry of the $D$-VASS $G_j$ are built into states when constructing $G_j'$.
We may assume that $k\in[c]$, otherwise we may swap $\mathbf{r}_k$ for some $\mathbf{r}_i$ in $\{\mathbf{r}_1,\ldots,\mathbf{r}_c\}$ without sacrificing the linear independence of the $c$ rows. 
Assume that $d_{G_j'}=d_{G_j}=c$.
Let $\mathbf{O}$ be the $(D{\times}c)$-matrix whose linearly independent columns $\mathbf{u}_1,\ldots,\mathbf{u}_c$ are displacements of $c$ simple circles in $G_j'$. Notice that a state in $G_j'$ is a pair consisting of a state in $G_j$ and a nonnegative integer of polynomial size.
A simple cycle in $G_j'$ is composed of a multi-set of simple cycles in $G_j$. 
Consequently there exists a $(c{\times}c)$-matrix $\mathbf{R}_1$ such that
$$
    \mathbf{O}=\begin{bmatrix}
        \mathbf{u}_1 \cdots \mathbf{u}_c
    \end{bmatrix}=\begin{bmatrix}
        \mathbf{c}_1  \cdots  \mathbf{c}_c
    \end{bmatrix}\cdot\mathbf{R}_1.
$$
Let the matrix $\mathbf{O}_1$ be composed of the first $c$ rows of $\mathbf{O}$, and $\mathbf{O}_2$ be composed of the other rows of $\mathbf{O}$. 
Then
$$
    \mathbf{O}=\begin{bmatrix}
    \mathbf{O}_1\\\mathbf{O}_2
\end{bmatrix}=\begin{bmatrix}
        \mathbf{c}_1  \cdots  \mathbf{c}_c
    \end{bmatrix}\cdot \mathbf{R}_1=\begin{bmatrix}
        \mathbf{I}_c\\ \mathbf{R}_0
    \end{bmatrix}\cdot \begin{bmatrix}
        \mathbf{r}_1 \\
        \vdots\\
        \mathbf{r}_c
    \end{bmatrix}\cdot \mathbf{R}_1,
$$
which implies that
$$
    \mathbf{O}_1=\mathbf{I}_c\begin{bmatrix}
        \mathbf{r}_1 \\
        \vdots\\
        \mathbf{r}_c
    \end{bmatrix}\mathbf{R}_1= \begin{bmatrix}
        \mathbf{r}_1 \\
        \vdots\\
        \mathbf{r}_c
    \end{bmatrix}\mathbf{R}_1,\quad \mathbf{O}_2=\mathbf{R}_0\begin{bmatrix}
        \mathbf{r}_1 \\
        \vdots\\
        \mathbf{r}_c
    \end{bmatrix}\mathbf{R}_1.
$$
It follows that $\mathbf{O}_2=\mathbf{R}_0\begin{bmatrix}
        \mathbf{r}_1 \\
        \vdots\\
        \mathbf{r}_c
    \end{bmatrix}\mathbf{R}_1=\mathbf{R}_0\mathbf{O}_1$. Therefore
$$
    rank(\mathbf{O})\leq rank(\mathbf{O}_1)\leq c-1.
$$
The last inequality is valid because the $k$-th row of $\mathbf{O}$ is $\mathbf{0}$ due to Lemma~\ref{2024-09-04}. 
So the assumption $d_{G_j'}=d_{G_j}$ must be false.
\end{proof}

Lemma~\ref{decomp-2023-05-11} and Lemma~\ref{2024-01-17} imply that the length of the refinement sequence in~(\ref{2024-10-28}) is bounded by $d$, that is $k\le d$.

\subsection{Termination of Decomposition}
\label{sec-Leaf}

Studies in the geometrically $d$-dimensional VASSes ought to begin with the geometrically $2$-dimensional VASSes.
Let us recall a major result proved in~\cite{FuYangZheng2024} about the geometrically $2$-dimensional VASSes.
\begin{theorem}\label{theorem-d-vasslps}
If there is a witness to $p(\mathbf{a})Gq(\mathbf{b})$, where $G=(Q,T)$ is a geometrically $2$-dimensional VASS, there is a witness to $\mathbf{a}\lambda\mathbf{b}$ where (i) $\lambda$ is a geometrically $2$-dimensional LCGS, (ii) $\lambda$ refines $pGq$, and (iii) $|\lambda|=2^{O(|G|)}$.
\end{theorem}
Theorem~\ref{theorem-d-vasslps} is proved in~\cite{FuYangZheng2024}.
In Section~\ref{Sec-Geometrically-2-Dimensional-d-VASS} we shall give a different and simpler proof of Theorem~\ref{theorem-d-vasslps}.
A consequence of Theorem~\ref{theorem-d-vasslps} is the following corollary that is significant to the design of our version of the KLMST algorithm.
\begin{corollary}\label{2023-12-31-xh}
If $\mathcal{W}_{\mathbf{a}\xi\mathbf{b}}\ne\emptyset$ and $\xi_j$ is geometrically $2$-dimensional, then $\xi_j$ can be replaced by a geometrically $2$-dimensional LCGS $\lambda_j$ of size $2^{O(|\xi_j|)}$ rendering true $\mathcal{W}_{\mathbf{a}\xi'\mathbf{b}}\ne\emptyset$, where $\xi'$ is obtained from $\xi$ by substituting $\lambda_j$ for $\xi_j$. Moreover $\xi'\sqsubseteq\xi$.
\end{corollary}
\begin{proof}
Suppose $\mathbf{w}\in\mathcal{W}_{\mathbf{a}\xi\mathbf{b}}$ and  $\xi_j=p_jG_jq_j$.
Let $\mathbf{w}_j$ be the sub-walk of $\mathbf{w}$ from $p_j(\mathbf{a}_j)$ to $q_j(\mathbf{b}_j)$ inside $G_j$.
By Theorem~\ref{theorem-d-vasslps}, there is an LCGS $\lambda_j$ whose size is bounded by $2^{O(|\xi_j|)}$ such that $\lambda_j\sqsubseteq\xi_j$ and there is a witness $\mathbf{w}_j'$ to $\mathbf{a}_j\lambda_j\mathbf{b}_j$.
Let $\mathbf{w}'$ be obtained from $\mathbf{w}$ by substituting $\mathbf{w}_j'$ for $\mathbf{w}_j$.
Then $\mathbf{w}'\in\mathcal{W}_{\mathbf{a}\xi'\mathbf{b}}$, where $\xi'$ is obtained from $\xi$ by substituting $\lambda_j$ for $\xi_j$.
By construction $\xi'\sqsubseteq\xi$.
\end{proof}

Since the existence of a witness to an LCGS can be checked algebraically, Corollary~\ref{2023-12-31-xh} points out that the sequence in~(\ref{2024-10-28}) may stop at a geometrically $2$-dimensional CG, hence the tighter bound $k\le d-2$.

It should be emphasized that the exponential blowup stated in Corollary~\ref{2023-12-31-xh} is local.
The exponential is about $|\xi_j|$, not about $|\xi|$.
This is in contrast to the global exponential blowup incurred by the decomposition operations. 

\subsection{Size of Decomposition}
\label{sec-SuanFa}

We are ready to define our version of the KLMST decomposition algorithm that operates on CGSes.
Let $\mathbf{a}\Xi\mathbf{b}$ denote the current CGS.
Initially it is the input $\mathbf{a}\xi\mathbf{b}$.
Here is the outline of the algorithm:
\begin{enumerate}
\item \label{step1}
Apply the Eulerian simplification and the orthogonal simplification to each component of the current CGS.
\item If $\Xi_j$ is geometrically $2$-dimensional, guess an LCGS $\lambda_j$ of size $2^{O(|\Xi_j|)}$ and substitute $\lambda_j$ for $\Xi_j$.
\item If $\Xi_j$ is not geometrically $2$-dimensional, nondeterministically choose a minimal solution $\mathbf{m}\in\mathcal{H}(\mathcal{E}_{\mathbf{a}\Xi\mathbf{b}})$.
\begin{enumerate}
\item If $\Xi_j$ is bounded, construct an algebraic decomposition of $\Xi_j$ using $\mathbf{m}$, and then substitute it for $\Xi_j$.
\item
    If $\Xi_j$ is unbounded and not $\mathbf{m}$-pumpable, construct a combinatorial decomposition of $\Xi_j$ and substitute it for $\Xi_j$.
\end{enumerate}
\item If $\mathbf{a}\Xi\mathbf{b}$ is normal, output `yes' and terminate; if not, go to Step~\ref{step1}.
\end{enumerate}
The algorithm aborts whenever the current CGS is unsatisfiable.
A nondeterministic execution of the KLMST algorithm is {\em successful} if it outputs a normal CGS.
\begin{proposition}\label{2024-01-01-first}
$\mathcal{W}_{\mathbf{a}\xi\mathbf{b}}\ne\emptyset$ if and only if upon receiving the input $\mathbf{a}\xi\mathbf{b}$ the KLMST algorithm has a successful execution.
\end{proposition}
\begin{proof}
If $\mathcal{W}_{\mathbf{a}\xi\mathbf{b}}\ne\emptyset$,  Lemma~\ref{Eulerian-Splf}, Lemma~\ref{WeqWp}, Lemma~\ref{my-decomposition-lemma}, Lemma~\ref{my-reduction-lemma} and Corollary~\ref{2023-12-31-xh} imply that there is a successful execution.
If $\mathcal{W}_{\mathbf{a}\xi\mathbf{b}}=\emptyset$, the algorithm must abort in all executions.
To prove the claim, assume that there were a successful execution that produces a normal CGS $\Xi$.
The execution could have guessed wrong minimal solutions.
But the point is that the CGS sequence the execution have generated is a sequence of successive refinements.
By Proposition~\ref{2024-11-01}, a witness to $\mathbf{a}\Xi\mathbf{b}$ can be constructed, which is also a witness to $\mathbf{a}\xi\mathbf{b}$ because $\Xi$ is a refinement of $\xi$.
This is in contradiction to the assumption $\mathcal{W}_{\mathbf{a}\xi\mathbf{b}}=\emptyset$.
\end{proof}

We need to know the output size of a successful execution.
The following theorem provides an answer, where $F_d$ and is defined in Section~\ref{s-Non-Elementary-Complexity-Class}.
The proof of the theorem is placed in Section~\ref{Sec-Proliferation-Tree}.
\begin{theorem}\label{OUR-F-D}
Upon receiving a geometrically $d$-dimensional input instance $\mathbf{a}\xi\mathbf{b}$, a successful execution of the KLMST algorithm produces some $\mathbf{a}\Xi\mathbf{b}\in g\mathbb{VASS}^d$ such that  $|\mathbf{a}\Xi\mathbf{b}|<F_d(|\mathbf{a}\xi\mathbf{b}|)$.
\end{theorem}

In the light of Proposition~\ref{2024-01-01-first} and Theorem~\ref{OUR-F-D}, our algorithm for $g\mathbb{VASS}^d$ is very simple: 
Upon receiving a reachability instance $\mathbf{a}G\mathbf{b}$, guess a walk from $p(\mathbf{a})$ to $q(\mathbf{b})$ whose length is bounded by $F_d(|\mathbf{a}G\mathbf{b}|)$.
The proof of Theorem~\ref{MAIN} is complete.

\section{Geometrically $2$-Dimensional VASS}
\label{Sec-Geometrically-2-Dimensional-d-VASS}

We now prove that every geometrically $2$-dimensional VASS can be refined to a geometrically $2$-dimensional LCGS of exponential size.
Let $G=(Q,T)$ be a geometrically $2$-dimensional $D$-VASS.
Notice that
\begin{eqnarray*}
|Q|,|T|,\|T\|\le 2^{O(|G|)}. \end{eqnarray*}
For every VASS $G'=(Q',T')$ introduced in the rest of the section, we will maintain the following complexity bounds.
\begin{eqnarray}\label{2024-01-newbound}
|Q'|,|T'|,\|T'\|\le 2^{O(|G|)}.
\end{eqnarray}
The bound $2^{O(|G|)}$ in Theorem~\ref{theorem-d-vasslps} is essentially due to~(\ref{2024-01-newbound}).

Our proof follows the line of investigation advocated in~\cite{BlondinFinkelGoellerHaaseMcKenzie2015,FuYangZheng2024}.
It differs from those of~\cite{FuYangZheng2024} in technical terms.
In particular it makes use of a simpler technique to project a geometrically $2$-dimensional walk to a $2$-dimensional walk.

\subsection{Circular Paths in a Geometrically  $2$-Dimensional Space}
\label{sec-DR}

Let $\mathbf{C}_G$ be the set of the simple cycles of $G$.
From now on we see $V_G$ as the subspace of $\mathbb{Q}^D$ spanned by $\mathbf{C}_G$.
Recall that a subspace passes through $\mathbf{0}$.
Let $\mathbf{r},\mathbf{s}$ be locations in $\mathbb{Z}^D$, and consider a circular path $q(\mathbf{r})\stackrel{\O}{\longrightarrow}_{\mathbb{Z}^d}q(\mathbf{s})$.
We are going to take a look at how to project this path onto a two dimensional space.
Let $\mathbf{c}=\mathbf{s}-\mathbf{r}$.
Let $\bot_{\mathbf{c}}\subseteq[D]$ be the set of $i$ such that $\mathbf{c}(i)=0$.
The vector $\mathbf{c}$ is either collinear with $\Delta(\o)$ for some $\o\in\mathbf{C}_G$, or is a linear combination of $\Delta(\o_1)$ and $\Delta(\o_2)$, where $\o_1,\o_2\in\mathbf{C}_G$ and $\Delta(\o_1),\Delta(\o_2)$ span $V_G$.
In the former case
\begin{equation}\label{2024-01-11-night}
\frac{|c_i|}{|c_k|}= \frac{|\Delta(\o)(i)|}{|\Delta(\o)(k)|}\le|Q|{\cdot}\|T\|
\end{equation}
for all $i,k\in[D]\setminus\bot_{\mathbf{c}}$.
In the latter case,
\begin{eqnarray*}
\mathbf{c} &=& \ell_1\Delta(\o_1)+\ell_2\Delta(\o_2)
\end{eqnarray*}
for some nonzero rational numbers $\ell_1,\ell_2$.
Let $\mathbf{u}'$ be $\Delta(\o_1)$ if $\ell_1>0$ and be $-\Delta(\o_1)$ otherwise.
Similarly let $\mathbf{v}'$ be $\Delta(\o_2)$ if $\ell_2>0$ and be $-\Delta(\o_2)$ otherwise.
Clearly
$\|\mathbf{u}'\|,\|\mathbf{v}'\| \,\le\, |Q|{\cdot}\|T\|$ and
\begin{eqnarray}\label{2023-11-08}
\mathbf{c} &=& |\ell_1|{\cdot}\mathbf{u}'+|\ell_2|{\cdot}\mathbf{v}'.
\end{eqnarray}
The vectors $\mathbf{u}',\mathbf{v}'$ may not be in $Z_{\mathbf{c}}=Z_{\#_1,\ldots,\#_D}$.
We will replace them by those in $Z_{\mathbf{c}}$.
The latter are obtained by solving the homogeneous equation system
\begin{eqnarray}\label{2024-01-11-evening}
\begin{pmatrix}
\#_1 z_1 \\
\vdots \\
\#_D z_D
\end{pmatrix} &=& z'\mathbf{u}'+z''\mathbf{v}'.
\end{eqnarray}
Since $\mathbf{u}',\mathbf{v}'$ are not collinear, the vectors $(z_1,\ldots,z_D)^T$ we are looking for are nontrivial.
It follows from~(\ref{2023-11-08}) that $m'\mathbf{c},k',l'$ form a solution to~(\ref{2024-01-11-evening}) for some positive integers $m',k',l'$.
The solution must be a multi-set of the minimal solutions to~(\ref{2024-01-11-evening}).
In other words, $\mathbf{c}=\sum_{i\in[l]}\ell^i\mathbf{m}_i$, where $\ell^1,\ldots,\ell^l$ are positive rational numbers, and $\mathbf{m}_1,\ldots,\mathbf{m}_l$ are the minimal solutions restricted to the left part of~(\ref{2024-01-11-evening}).
Because $V_G$ is geometrically $2$-dimensional and $\mathbf{m}_1,\ldots,\mathbf{m}_l$ are in the same orthant, we derive from $\mathbf{c}=\sum_{i\in[l]}\ell^i\mathbf{m}_i$ that either $\mathbf{c}=\ell\mathbf{u}$ for some positive rational number $\ell$ and some $\mathbf{u}\in\{\mathbf{m}_1,\ldots,\mathbf{m}_l\}$, or $\mathbf{c}=\ell\mathbf{u}+\ell'\mathbf{v}$ for some positive rational numbers $\ell,\ell'$ and some $\mathbf{u},\mathbf{v}\in\{\mathbf{m}_1,\ldots,\mathbf{m}_l\}$.
The former case has been treated in the beginning of this paragraph.
In the other case, notice that
$\|\mathbf{u}'\|,\|\mathbf{v}'\| \,\le\, |Q|{\cdot}\|T\|$ and Lemma~\ref{pottier-lemma} together imply that
\begin{equation}\label{2023-10-14}
\|\mathbf{u}\|,\|\mathbf{v}\| \,\le\, (|Q|{\cdot}\|T\|)^{O(1)}.
\end{equation}
We are going to project $\mathbf{c}$ onto two of its entries so that we can use Lemma~\ref{2023-12-27}.
The projection is defined by fixing two linearly independent pairs $(u_{\imath},v_{\imath})$, $(u_{\jmath},v_{\jmath})$, where
\begin{equation}\label{2024-08-27}
\{\imath,\jmath\}\cap\bot_{\mathbf{c}}=\emptyset.
\end{equation}
The pair are picked up by the following strategy:
\begin{enumerate}
\item If $u_{\imath}=0$ for some $\imath\in[D]\setminus\bot_{\mathbf{c}}$ and $v_{\jmath}=0$ for some $\jmath\in[D]\setminus\bot_{\mathbf{c}}$, then pick the pairs $(u_{\imath},v_{\imath})$ and $(u_{\jmath},v_{\jmath})$.
Notice that $u_{\imath}=0$ and $v_{\imath}=0$ cannot happen at the same time, otherwise $\mathbf{c}(\imath)=0$, contradicting to~(\ref{2024-08-27}).
For the same reason, $u_{\jmath}=0$ and $v_{\jmath}=0$ cannot happen at the same time.
\item If $u_{\imath}=0$ for some $\imath\in[D]\setminus\bot_{\mathbf{c}}$ and $v_{\jmath}\ne0$ for all $\jmath\in[D]\setminus\bot_{\mathbf{c}}$, then pick the pair $(u_{\imath},v_{\imath})$ and pick some $(u_{\jmath},v_{\jmath})$ such that $u_{\jmath}\ne0$.
    This is possible because $\mathbf{u}\ne\mathbf{0}$.
\item If $u_{\imath}\ne0$ for all $\imath\in[D]\setminus\bot_{\mathbf{c}}$ and $v_{\jmath}=0$ for some $\jmath\in[D]\setminus\bot_{\mathbf{c}}$, then pick the pair $(u_{\jmath},v_{\jmath})$ and pick some  $(u_{\imath},v_{\imath})$ such that $v_{\imath}\ne0$.
    This is possible because $\mathbf{v}\ne\mathbf{0}$.
\item If $u_{\imath}\ne0$ for all $\imath\in[D]\setminus\bot_{\mathbf{c}}$ and $v_{\jmath}\ne0$ for all $\jmath\in[D]\setminus\bot_{\mathbf{c}}$, then pick up linearly independent $(u_{\imath},v_{\imath})$ and $(u_{\jmath},v_{\jmath})$.
    This is possible because $\mathbf{u}$ and $\mathbf{v}$ are linearly independent.
\end{enumerate}
Now $c_i = \ell u_{i}+\ell'v_{i}$ and $|c_i| = \ell|u_{i}|+\ell'|v_{i}|$ for all $i\in[D]$.
The second equality is due to the fact that $\mathbf{u},\mathbf{v}$ are in the same orthant.
Let $k\in[D]\setminus\bot_{\mathbf{c}}\setminus\{\imath,\jmath\}$.
If $v_k=0$, then $v_{\jmath}=0$, and
\begin{equation}\label{2024-1-11}
\frac{|c_{\jmath}|}{|c_k|}= \frac{|u_{\jmath}|}{|u_k|} \le (|Q|{\cdot}\|T\|)^{O(1)},
\end{equation}
where the upper bound is from~(\ref{2023-10-14}).
If $u_k=0$, then $u_{\imath}=0$, and
\begin{equation}\label{2024-1-11-b}
\frac{|c_{\imath}|}{|c_k|}= \frac{|v_{\imath}|}{|v_k|} \le (|Q|{\cdot}\|T\|)^{O(1)}.
\end{equation}
If $u_k\ne0$ and $v_k\ne0$, then
\begin{equation}\label{2023-10-15}
\frac{|c_{\jmath}|}{|c_k|}=\frac{\ell{\cdot}|u_{\jmath}|+\ell'{\cdot}|v_{\jmath}|}{\ell{\cdot}|u_k|+\ell'{\cdot}|v_k|} \le \max\left\{\frac{|u_{\jmath}|}{|u_k|},\frac{|v_{\jmath}|}{|v_k|}\right\} \le (|Q|{\cdot}\|T\|)^{O(1)}.
\end{equation}
We remark that the upper bound in these inequalities is polynomially related to the upper bound in~(\ref{2024-01-11-night}).

\subsection{Dividing $\mathbb{N}^D$ into $\mathbb{L}_G$ and $\mathbb{U}_G$}
\label{sec-CP}

Suppose $q(\mathbf{a})\stackrel{\O}{\longrightarrow}_{\mathbb{N}^D}q(\mathbf{b})$.
Let $\mathbf{c}=\mathbf{b}-\mathbf{a}$.
Using the notations in Section~\ref{sec-DR}, we project $q(\mathbf{a})\stackrel{\O}{\longrightarrow}_{\mathbb{N}^D}q(\mathbf{b})$ onto the $2$-dimensional walk $q(\mathbf{a}(\imath),\mathbf{a}(\jmath))\stackrel{\O^{\imath,\jmath}}{\longrightarrow}_{\mathbb{N}^2}q(\mathbf{b}(\imath),\mathbf{b}(\jmath))$, where $\O^{\imath,\jmath}$ is the projection of $\O$ onto the $\imath$-th and the $\jmath$-th coordinates.
By Lemma~\ref{2023-12-27} there is a $2$-dimensional LCGS $\lambda$ of size $2^{O(|G|)}$ such that
(i) $\lambda$ refines $qGq$ and (ii) $\lambda$ contains at most two cycles $\O_1,\O_2$ and $\Delta(\O_1),\Delta(\O_2)\in Z_{\mathbf{b}-\mathbf{a}}$ and
(iii) some $q(\mathbf{a}(\imath),\mathbf{a}(\jmath))\stackrel{\pi}{\longrightarrow}_{\mathbb{Z}^2}q(\mathbf{b}(\imath),\mathbf{b}(\jmath))$ is admitted by $\lambda$.
We restore from $\pi$ a $D$-dimensional path
$q(\mathbf{a})\stackrel{\varpi}{\longrightarrow}_{\mathbb{Z}^D}q(\mathbf{b}')$
such that $\varpi^{\imath,\jmath}=\pi$.
We remark that $\varpi$ might not be unique since two distinct transitions of $G$ might coincide after the projection.
But since $q(\mathbf{a})\stackrel{\O}{\longrightarrow}_{\mathbb{N}^D}q(\mathbf{b})$ and $q(\mathbf{a})\stackrel{\varpi}{\longrightarrow}_{\mathbb{Z}^D}q(\mathbf{b}')$ are circular walks, $\mathbf{a},\mathbf{b},\mathbf{b}'$ are in the same $2$-dimensional plane.
It follows that, for every $i\in[D]\setminus\bot_{\mathbf{c}}\setminus\{\imath,\jmath\}$, the equality $\mathbf{b}(i)=\mathbf{b}'(i)$ is implied by $\mathbf{b}'(\imath)=\mathbf{b}(\imath)$ and $\mathbf{b}'(\jmath)=\mathbf{b}(\jmath)$.
Also notice that $\mathbf{b}'(i)=\mathbf{a}(i)=\mathbf{b}(i)$ for every $i\in\bot_{\mathbf{c}}$.
Hence $\mathbf{b}'=\mathbf{b}$.

The size of $\lambda$ is bounded by $2^{O(|G|)}$.
If $\mathbf{a}(i),\mathbf{b}(i)\ge|\lambda|{\cdot}\|T\|=2^{O(|G|)}{\cdot}\|T\|$ for all $i\in[D]$, the path $q(\mathbf{a})\stackrel{\varpi}{\longrightarrow}_{\mathbb{Z}^D}q(\mathbf{b})$ is an $\{\imath,\jmath\}$-walk.
It is not guaranteed to be in the first orthant though.
But in the light of~(\ref{2024-01-11-night}), (\ref{2024-1-11}), (\ref{2024-1-11-b}) and~(\ref{2023-10-15}),
if $q(\mathbf{a})\stackrel{\varpi}{\longrightarrow}_{\mathbb{Z}^D}q(\mathbf{b})$ passes a location $\mathbf{d}$ such that $0\le d_k<\|T\|$ for some $k\in[D]\setminus\bot_{\mathbf{c}}\setminus\{\imath,\jmath\}$, then either $d_{\imath}<(|Q|{\cdot}\|T\|)^{O(1)}{\cdot}\|T\|$ or $d_{\jmath}<(|Q|{\cdot}\|T\|)^{O(1)}{\cdot}\|T\|$.
Consequently if for all $i\in[D]$,
\begin{equation}\label{2024-08-29}
|\mathbf{a}(i)|,|\mathbf{b}(i)|\ge B_G \stackrel{\rm def}{=} 2^{O(|G|)}{\cdot}\|T\| +(|Q|{\cdot}\|T\|)^{O(1)}{\cdot}\|T\|=2^{O(|G|)},
\end{equation}
then $q(\mathbf{a})\stackrel{\varpi}{\longrightarrow}_{\mathbb{N}^D}q(\mathbf{b})$.
The parameter $B_G$ tells us to divide the first orthant into two parts, one way up in the space, the other near to the axis planes.
Let us set
\begin{eqnarray*}
\mathbb{U}_G &=& [B_G,\infty)^D, \\
\mathbb{L}_{G} &=& \bigcup_{i\in[D]}\mathbb{I}_i,
\end{eqnarray*}
where for each $i\in[D]$,
\begin{eqnarray*}
\mathbb{I}_i &=& \mathbb{N}^{i-1}\times[0,B_G)\times\mathbb{N}^{D-i}.
\end{eqnarray*}
A walk is then divided to a sequence of sub-walks whose locations belong to either $\mathbb{L}_{G}$ or $\mathbb{U}_G$ in an alternating manner.
These sub-walks are connected by edges that cross between $\mathbb{L}_{G}$ and $\mathbb{U}_G$.
Confer the diagram in Figure~\ref{borderline}.
We are going to prove by induction that every sub-walk completely inside $\mathbb{U}_G$ or completely inside $\mathbb{L}_G$ can be admitted by LCGSes of small size.

\subsection{Dividing Walks by the Border between $\mathbb{L}_G$ and $\mathbb{U}_G$}
\label{sec-W-U-L}

Consider $p(\mathbf{a})\longrightarrow^{*}_{\mathbb{N}^D}q(\mathbf{b})$.
For each $i$ orthogonal to $G$, we carry out the orthogonal simplification, after which the number of vertices is below the bound in~(\ref{2024-01-newbound}), and the $i$-th entry can be ignored.
So we only have to consider the indices that are not orthogonal to $G$.

\begin{figure}[t]
\centering
\includegraphics[scale=0.8]{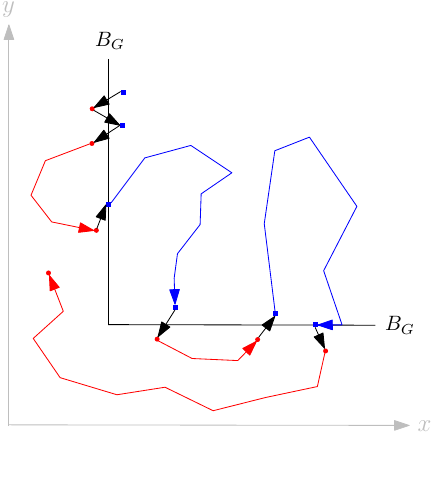}
\caption{A Walk Divided by the Border between $\mathbb{L}_G$ and $\mathbb{U}_G$.}  \label{borderline}
\end{figure}

Let $\mathbb{U}_{\rm cfg}$ be the multi-set of configurations in $p(\mathbf{a})\longrightarrow^{*}_{\mathbb{N}^D}q(\mathbf{b})$ that are also in $\mathbb{U}_G$.
Let $p_1(\mathbf{a}_1),p_2(\mathbf{a}_2),\ldots,p_b(\mathbf{a}_b)$ be the elements of $\mathbb{U}_{\rm cfg}$, listed in the order that they appear in $p(\mathbf{a})\longrightarrow^{*}_{\mathbb{N}^D}q(\mathbf{b})$.
If $b<|Q|$, the total length of the sub-walks of $p(\mathbf{a})\longrightarrow^{*}_{\mathbb{N}^D}q(\mathbf{b})$ in $\mathbb{U}_G$ is short. 
If $b\ge|Q|$, we look for two configurations $p_i(\mathbf{a}_i),p_{i'}(\mathbf{a}_{i'})$, whose existence is guaranteed, rendering true the following:
\begin{itemize}
\item $i<i'$ and $p_i=p_{i'}$;
\item $p_i\notin\{p_1,\ldots,p_{i-1}\}$ and $p_{i'}\notin\{p_{i'+1},\ldots,p_b\}$; and
\item $\{p_1,\ldots,p_{i-1}\}\cap\{p_{i'+1},\ldots,p_b\}=\emptyset$.
\end{itemize}
Clearly $p_i(\mathbf{a}_i)\longrightarrow^{*}_{\mathbb{N}^D}p_{i'}(\mathbf{a}_{i'})$ is a sub-walk of $p(\mathbf{a})\longrightarrow^{*}_{\mathbb{N}^D}q(\mathbf{b})$ that is circular and $\mathbf{a}_i,\mathbf{a}_{i'}\in\mathbb{U}_G$.
By applying the argument inductively to $p(\mathbf{a})\longrightarrow^{*}_{\mathbb{N}^D}p_{i-1}(\mathbf{a}_{i-1})$ and $p_{i'+1}(\mathbf{a}_{i'+1})\longrightarrow^{*}_{\mathbb{N}^D}q(\mathbf{b})$ we conclude that $p(\mathbf{a})\longrightarrow^{*}_{\mathbb{N}^D}q(\mathbf{b})$ must be of the following form for some $c<|Q|$.
\begin{equation}\label{2024-02-04}
p(\mathbf{a})\xrightarrow{\pi_1 \alpha_1\varpi_1 \alpha_2\pi_2\varpi_2\ldots \alpha_c\pi_c\varpi_c \alpha_{c+1}\pi_{c+1}}_{\mathbb{N}^D}q(\mathbf{b}).
\end{equation}
In~(\ref{2024-02-04}), $\varpi_1,\varpi_2\ldots,\varpi_c$ are circular walks whose first and last configurations are in $\mathbb{U}_G$.
Suppose $\varpi_i$ is $p_i(\mathbf{a}_i)\stackrel{\varpi_i}{\longrightarrow}_{\mathbb{N}^D}p_{i'}(\mathbf{a}_{i'})$ for $i\in[c]$.
By Lemma~\ref{2023-12-27} and the correctness of the projection method discussed in Section~\ref{sec-CP},
$p_i(\mathbf{a}_i)\longrightarrow^{*}_{\mathbb{N}^D}p_{i'}(\mathbf{a}_{i'})$ is admitted by some LCGS $\lambda_{i,i'}$ that refines $p_{i}Gp_{i'}$ and contains at most two cycles $\O_1,\O_2$ such that $\Delta(\O_1),\Delta(\O_2)\in Z_{\mathbf{a}_{i'}-\mathbf{a}_{i}}$; moreover $|\lambda_{i,i'}|=2^{O(|G|)}$.
Also in~(\ref{2024-02-04}), 
$\alpha_1, \alpha_2\ldots, \alpha_{c+1}$ are short walks crossing between $\mathbb{U}_G$ and $\mathbb{L}_G$, which are among the black arrows in Figure~\ref{borderline}.
Notice that $|\alpha_1|+|\alpha_2|+\ldots+|\alpha_{c+1}|\le2|Q|$, otherwise there would be a cycle whose initial and final locations are in $\mathbb{U}_G$, which would be one of $\varpi_1,\varpi_2\ldots,\varpi_c$.
Finally all the walks $\pi_1,\pi_2\ldots,\pi_{c+1}$ are completely in $\mathbb{L}_G$.
It follows from the above discussion that for each $i\in[c+1]$, the sub-walk $\pi_i$ must be of the form $\pi_i^0\O_i^1\pi_i^1\ldots\pi_i^{k_i}\O_i^{k_i}\pi_i^{k_i+1}$, where $\O_i^1,\ldots,\O_i^{k_i}$ are circular and $|\pi_i^0\pi_i^1\ldots\pi_i^{k_i}\pi_i^{k_i+1}|<|Q|$.
We will show in Section~\ref{sec-Walks-in-LG} that every circular walk in $\mathbb{L}_G$ is admitted by a $2^{O(|G|)}$-size LCGS, which will complete the proof of Theorem~\ref{theorem-d-vasslps}.

\subsection{Walks in $\mathbb{L}_G$}
\label{sec-Walks-in-LG}

Suppose $o(\mathbf{c})\longrightarrow_{\mathbb{L}_G}^{*}o(\mathbf{c}')$.
We introduce, for distinct $i,j\in[D]$, the {\em crossing region} between $\mathbb{I}_i$ and $\mathbb{I}_j$ defined as follows:
\begin{eqnarray*}
\mathbb{J}_{ij} &\stackrel{\rm def}{=}& \mathbb{I}_i\cap\mathbb{I}_j.
\end{eqnarray*}
Let $D'=D-|\bot_G|$, recall that $\bot_G$ is the set of the indices orthogonal to $G$.
We shall only be concerned with the crossing region $\mathbb{J}_{ij}$ where neither $i$ nor $j$ being orthogonal to $G$.
There are altogether $\frac{1}{2}D'(D'\,{-}\,1)$ such crossing regions.
A walk that starts from some location in $\mathbb{I}_i\setminus\mathbb{I}_j$ and goes to a location in $\mathbb{I}_j\setminus\mathbb{I}_i$, or vice versa, may pass through the crossing region $\mathbb{J}_{ij}$ several times.
Now suppose there is a sub-walk $o_0(\mathbf{c}_0)\longrightarrow_{\mathbb{L}_G}^{*}o_1(\mathbf{c}_1)$ that passes at least $|Q|{\cdot}B_G^2{\cdot}\|T\|+1$ configurations in $\mathbb{J}_{ij}$.
By the pigeon hole principle, $o_0(\mathbf{c}_0)\longrightarrow_{\mathbb{L}_G}^{*}o_1(\mathbf{c}_1)$ must pass two configurations $o'(\mathbf{c}'),o''(\mathbf{c}'')$ such that $o'=o''$, $(\mathbf{c}''-\mathbf{c}')(i)=0$ and $(\mathbf{c}''-\mathbf{c}')(j)=0$.
In other words there is a circular walk $\O'$ whose displacement is orthogonal to both the $i$-th axis and the $j$-th axis.
In the presence of $\O'$, the number of configurations in $\mathbb{I}_{ik}$, for each $k\in[D]\setminus\bot_G\setminus\{i,j\}$, the walk $o(\mathbf{c})\longrightarrow_{\mathbb{L}_G}^{*}o(\mathbf{c}')$ may pass must be bounded by $|Q|{\cdot}B_G^2{\cdot}\|T\|$.
Otherwise by the same argument, there would be a circular walk $\O''$ whose displacement is orthogonal to both the $i$-th axis and the $k$-th axis, which would imply that the $2$-dimensional space $V_G$ were orthogonal to the $i$-th axis, contradicting to our assumption $i\notin\bot_G$.
Thus there can be two cases.
\begin{enumerate}
\item For each $k\in[D]\setminus\bot_G\setminus\{i\}$, the circular walk $o(\mathbf{c})\longrightarrow_{\mathbb{L}_G}^{*}o(\mathbf{c}')$ passes no more than $|Q|{\cdot}B_G^2{\cdot}\|T\|$ configurations in $\mathbb{J}_{ik}$.
    In this case we divide the first quadrant by the hyperplane $x_i=B_G$.
    \begin{itemize}
    \item
    In $\mathbb{I}_i$ there are at most $\frac{1}{2}{\cdot}(D'-1){\cdot}|Q|{\cdot}B_G^2{\cdot}\|T\|+1$ sub-walks between these configurations.
    By the construction described in Section~\ref{Sec-Eulerian-Simplification} right after Lemma~\ref{WeqWp}, we can turn $\mathbb{I}_i$ to a VASS to which  $i$ is orthogonal, which by the definition of $B_G$ in~(\ref{2024-08-29}), satisfies the bound in~(\ref{2024-01-newbound}).
    By induction each of the sub-walks in $\mathbb{I}_i$ is admitted by an LCGS of size $2^{O(|G|)}$.
    \item
    There are also at most $\frac{1}{2}{\cdot}(D'-1){\cdot}|Q|{\cdot}B_G^2{\cdot}\|T\|+1$ sub-walks in $\mathbb{L}_G\setminus\mathbb{I}_i$.
    There are altogether $\frac{1}{2}(D'\,{-}\,1)(D'\,{-}\,2)$ crossing regions in $\mathbb{L}_G\setminus\mathbb{I}_i$.
    By induction on the number of crossing regions, each of the sub-walks is admitted by an LCGS of size $2^{O(|G|)}$.
    \end{itemize}
    We conclude that $o(\mathbf{c})\longrightarrow_{\mathbb{L}_G}^{*}o(\mathbf{c}')$ is admitted by the concatenation of at most $(D'-1){\cdot}|Q|{\cdot}B_G^2{\cdot}\|T\|+2=2^{O(|G|)}$ LCGSes of size $2^{O(|G|)}$, the total size being $2^{O(|G|)}$.
    \item The circular walk $o(\mathbf{c})\longrightarrow_{\mathbb{L}_G}^{*}o(\mathbf{c}')$ passes more than $|Q|{\cdot}B_G^2{\cdot}\|T\|$ configurations in $\mathbb{J}_{ij}$, and for each $k\in[D]\setminus\bot_G\setminus\{i,j\}$, the circular walk $o(\mathbf{c})\longrightarrow_{\mathbb{L}_G}^{*}o(\mathbf{c}')$ passes no more than $|Q|{\cdot}B_G^2{\cdot}\|T\|$ configurations in $\mathbb{J}_{ik}$ and  no more than $|Q|{\cdot}B_G^2{\cdot}\|T\|$ configurations in $\mathbb{J}_{jk}$.
    We have proved in the above that there is a circular sub-walk $\O'$ of $o(\mathbf{c})\longrightarrow_{\mathbb{L}_G}^{*}o(\mathbf{c}')$ whose displacement is orthogonal to both the $i$-th axis and the $j$-th axis.
    Now $o(\mathbf{c})\longrightarrow_{\mathbb{L}_G}^{*}o(\mathbf{c}')$ contains at most $(D'-2){\cdot}|Q|{\cdot}B_G^2{\cdot}\|T\|+1$ sub-walks in $\mathbb{I}_i\cup\mathbb{I}_j$.
    Each of the sub-walks is either in $\mathbb{I}_i$ or in $\mathbb{I}_j$ or in both $\mathbb{I}_i$ and $\mathbb{I}_j$.
    The first and the second situations can be treated by induction.
    Assuming that $i<j$, a sub-walk in the third case is in \[\mathbb{N}^{i-1}\times\left[0,A_G{+}B_G\right)\times \mathbb{N}^{j-i-1}\times\left[0,A_G{+}B_G\right)\times\mathbb{N}^{D-j},\] where $A_G \stackrel{\rm def}{=} |Q|{\cdot}(B_G\,{+}\,1){\cdot}\|T\|$.
    To see this, suppose $o(\mathbf{c})\longrightarrow_{\mathbb{L}_G}^{*}o(\mathbf{c}')$ contains a sub-walk $q^f(\mathbf{c}^f)\stackrel{\pi}{\longrightarrow}_{\mathbb{I}_j}q^l(\mathbf{c}^l)$ such that
\begin{equation}\label{2023-12-26-noon}
|\mathbf{c}^l(j)-\mathbf{c}^f(j)| > |Q|{\cdot}(B_G\,{+}\,1){\cdot}\|T\|.
\end{equation}
It follows from~(\ref{2023-12-26-noon}) that in $\pi$ there must be at least $|Q|{\cdot}(B_G\,{+}\,1)$ configurations
\begin{equation}\label{2024-01-BD}
q_1(\mathbf{c}^1),q_2(\mathbf{c}^2),\ldots,q_{|Q|{\cdot}(B_G\,{+}\,1)}(\mathbf{c}^{|Q|{\cdot}(B_G\,{+}\,1)})
\end{equation}
rendering true the inequality $\mathbf{c}^{k+1}(j)-\mathbf{c}^k(j)\ge 1$
for all $k\in[|Q|{\cdot}(B_G\,{+}\,1)]$.
By the pigeonhole principle there are configurations $q'(\mathbf{c}^s),q''(\mathbf{c}^t)$ in~(\ref{2024-01-BD}) satisfying $q'=q''$, $\mathbf{c}^s(i)=\mathbf{c}^t(i)$ and $\mathbf{c}^s(j)\ne\mathbf{c}^t(j)$.
In other words, there is a circular sub-walk $\O''$ such that $\Delta(\O'')$ is orthogonal to the $i$-th axis but not orthogonal to the $j$-th axis.
Together with $\Delta(\O')$ it implies that $V_G$ is orthogonal to the $i$-th axis, contradicting to the assumption.
The rest of the argument is the same as in the above case.
\end{enumerate}
The proof of Theorem~\ref{theorem-d-vasslps} is now complete.

\section{Complexity Upper Bound}
\label{Sec-Proliferation-Tree}

The process of successive generations of CGSes in an execution of the KLMST algorithm can be visualized as tree growing, in which the leaves are the components of the CGS constructed so far.
The number of the CGes of the output CGS of a successful execution of KLMST is the number of the leaves of the final tree.
We shall formalize a form of tree growing and then use it to derive a complexity upper bound for our KLMST algorithm.

Fix an elementary function $p(n)=e(n)-1$, called a {\em proliferation function}.
The {\em height} of a finite tree is defined by the number of edges in the longest branch of the tree.
Suppose $\mathfrak{T}$ is a tree of height one whose root has $n$ children.
For $\hbar\ge3$ an $\hbar$-{\em proliferation} of $\mathfrak{T}$ is a tree of height $\hbar-1$ obtained by growing $\mathfrak{T}$.
Let us call a leaf in an $\hbar$-proliferation of $\mathfrak{T}$ {\em active} if its height is less than $\hbar-1$.
The rule of growth is this:
Let $\mathfrak{T}'$ be the tree half way through growing.
If the set of the active leaves is nonempty, choose one such a leaf and attach $2^{p(l)}$ children to the leaf, where $l$ is the number of the leaves of $\mathfrak{T}'$.
There is a huge number of $\hbar$-proliferations of $\mathfrak{T}$.
We are interested in the following question:
Which $\hbar$-proliferation is maximum in the sense that it has the maximum number of leaves?
It turns out that the answer to the question is standard.

The depth first strategy is to always grow an active leaf with the maximum height.
Our intuition tells us that to produce a maximum $\hbar$-proliferation tree of $\mathfrak{T}$, one should adopt the depth first strategy.

\begin{figure}[t]
\centering
\includegraphics[scale=0.35]{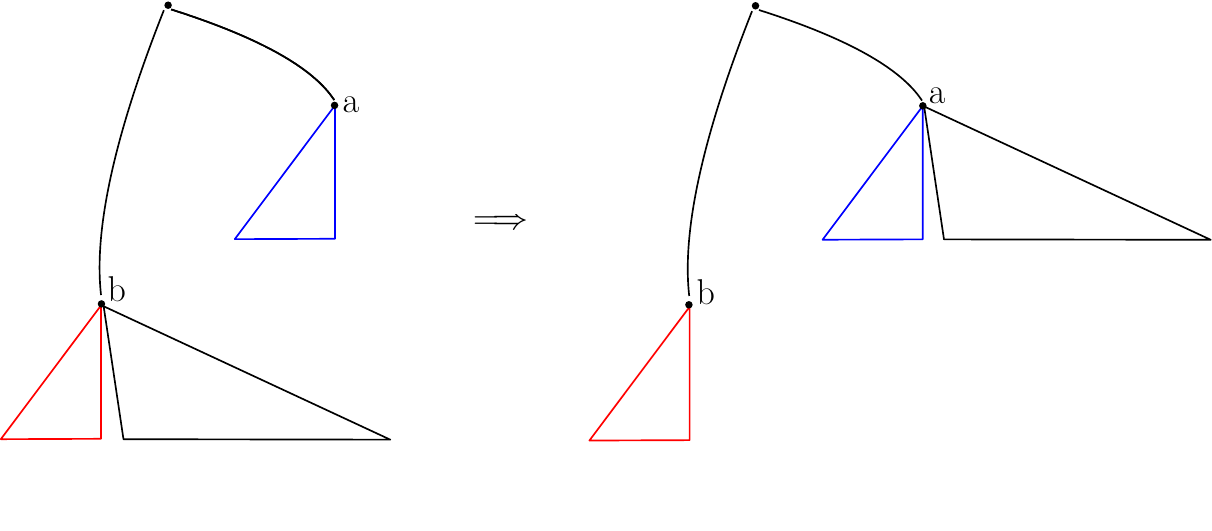}
\caption{Swapping the order of the growth of $a$ and $b$.}  \label{PT}
\end{figure}

\begin{lemma}
The depth first strategy always produces a maximum $\hbar$-proliferation of $\,\mathfrak{T}$.
\end{lemma}
\begin{proof}
Suppose in the generation of the $\hbar$-proliferation $\mathfrak{T}_0$ of $\mathfrak{T}$ an active leaf $a$ of height $h_0$ is grown in the $i$-th step and an active leaf $b$ of height $h_1$ is grown in the $(i{+}1)$-th step; moreover $h_0<h_1$.
Confer the left tree in Figure~\ref{PT}.
Suppose that before the $i$-th step is executed the tree has $l$ leaves.
Now grow a tree $\mathfrak{T}_1$ (the right tree in Figure~\ref{PT}) by simulating the growth of $\mathfrak{T}_0$ with the modification that $b$ is grown in the $i$-th step and that $a$ is grown in the $(i{+}1)$-th step.
The simulation is simple.
The blue (respectively red, black) subtree of the left tree is simulated by the blue (respectively red, black) subtree of the right tree.
Since $h_0<h_1$, at the time the simulation of the growth of $\mathfrak{T}_0$ is completed, the growth of the tree $\mathfrak{T}_1$ is not yet completed.
So the final tree $\mathfrak{T}_1$ must have more leaves than $\mathfrak{T}_0$.

If the growth of $\mathfrak{T}_1$ violates the depth first principle, we can apply the above conversion to $\mathfrak{T}_1$ while increasing the number of leaves.
Eventually we get some $\hbar$-proliferation of $\mathfrak{T}$ that is generated in the depth first fashion.
It is clear that if $h_0=h_1$, the conversion neither increases nor decreases the number of leaves.
It follows that all $\hbar$-proliferations of $\mathfrak{T}$ grown in the depth first manner have the same number of leaves.
\end{proof}
Next we are going to derive an upper bound on the number of the leaves of a maximum $\hbar$-proliferation of $\mathfrak{T}$ for $\hbar>3$.
The main technical result is the following, where $F_{\hbar}$ is the function defined in Section~\ref{s-Non-Elementary-Complexity-Class}.

\begin{lemma}\label{jinggangshan}
For each $\hbar>3$ and each tree $\mathfrak{T}$ of height $1$ with $n$ leaves, the number of the leaves of a maximum $\hbar$-proliferation of $\mathfrak{T}$ is bounded by $F_{\hbar}(n)$.
\end{lemma}
\begin{proof}
Let's now derive an upper bound for the maximum $(\hbar{+}1)$-proliferation of $\mathfrak{T}$ for $\hbar\ge3$.
In the depth first strategy the first child of the root is attached with $2^{p(n)}$ children.
Let $M_1=2^{p(n)}+n-1$.
Let $\mathfrak{T}^1$ denote the tree defined by the first child of $\mathfrak{T}$ and the $2^{p(n)}$ children.
By induction hypothesis the number of the leaves of the maximum $\hbar$-proliferation of $\mathfrak{T}^1$ is bounded by $F_{\hbar}(M_1)$.
In fact $F_{\hbar}(M_1)$ also bounds the number of the leaves of the entire tree right after the growth of the maximum $\hbar$-proliferation of $\mathfrak{T}^1$ is completed.
Suppose $k\in[n]\setminus[1]$.
For each $g\in[k]\setminus[1]$, let $\mathfrak{T}^{g}$ denote the tree defined by the $g$-th child of $\mathfrak{T}$ and its children.
Assume that the number of the leaves of the maximum $\hbar$-proliferation of $\mathfrak{T}^g$ is bounded by $F_{\hbar}(M_g)$, where $M_{g}=2^{p(F_{\hbar}(M_{g-1}))}+F_{\hbar}(M_{g-1})-1$.
Also assume that $F_{\hbar}(M_g)$ bounds the number of the leaves of the tree right after the growth of the maximum $\hbar$-proliferation of $\mathfrak{T}^g$ is completed.
Let $\mathfrak{T}^{k+1}$ denote the tree defined by the $(k{+}1)$-th child of $\mathfrak{T}$ and its children, which are less than $M_{k+1}=2^{p(F_{\hbar}(M_k))}+F_{\hbar}(M_k)-1\le 2^{e(F_{\hbar}(M_k))}$ in number.
By induction hypothesis the number of the leaves of the maximum $\hbar$-proliferation of $\mathfrak{T}^{k+1}$, as well as the number of the leaves of the tree right after the growth of the maximum $\hbar$-proliferation of $\mathfrak{T}^{k+1}$ is completed, is bounded by
\begin{eqnarray}
F_{\hbar}(M_{k+1}) &\le& F_{\hbar}\left(2^{e(F_{\hbar}(M_{k}))}\right) \nonumber \\
 &\le& F_{\hbar}\left(F_{\hbar}(M_{k}+N+1)\right) \label{2024-02-02} \\
 &\le& F_{\hbar}^2\left(2^{p(F_{\hbar}(M_{k-1}))}+F_{\hbar}(M_{k-1})-1+N+1\right) \nonumber \\
 &\le& F_{\hbar}^2\left(2^{e(F_{\hbar}(M_{k-1}))}\right) \nonumber \\
 &\le& \ldots \nonumber \\
 &\le& F_{\hbar}^k\left(F_{\hbar}(M_{1}+N+1)\right) \nonumber \\
 &=& F_{\hbar}^{k+1}\left(2^{p(n)}+n-1+N+1\right) \nonumber \\
 &\le& F_{\hbar}^{k+1}\left(2^{e(n)}\right) \nonumber \\
 &\le& F_{\hbar}^{k+2}(n). \label{2024-02-02-b}
\end{eqnarray}
The elementary function $e(n)$ must be bounded by $t(n,N)$ for some positive integer $N$.
In the above derivation, (\ref{2024-02-02}) and~(\ref{2024-02-02-b}) are valid because $F_{\hbar}$ grows faster than any elementary function.
It follows that the number of the leaves of the maximum $(\hbar{+}1)$-proliferation of $\mathfrak{T}$ is bounded by $F_{\hbar}(M_{n})\le F_{\hbar}^{n+1}(n)=F_{\hbar+1}(n)$.
\end{proof}

By Lemma~\ref{Eulerian-Splf},  Lemma~\ref{WeqWp}, Lemma~\ref{my-decomposition-lemma} and Lemma~\ref{my-reduction-lemma}, every time a refinement is performed to the current CGS $\mathbf{a}\Xi\mathbf{b}$, the size of the new CGS is non-increasing if the refinement is a simplification; it is bounded by $2^{\texttt{poly}(|\mathbf{a}\Xi\mathbf{b}|)}$ if the refinement is a decomposition.
If $\Xi_j=p_jG_jq_j$ is decomposed to a bunch of geometrically $2$-dimensional CGes, then according to Theorem~\ref{theorem-d-vasslps}, each of these CGes is replaced by an LCGS with a local exponential in size (see the remark in the last paragraph of Section~\ref{sec-Leaf}).
If we set $p(n)=2^{2^{2^{n}}}-1$, we get a bound on the number of the leaves of the KLMST decomposition tree.
It is also a bound on the size of the output normal CGS.

It is now clear that Theorem~\ref{OUR-F-D} follows from Lemma~\ref{jinggangshan}.

\section{Future Work}\label{Sec-Conclusion}

Theorem~\ref{MAIN} is a strong indication that the introduction of $g\mathbb{VASS}^d$ is on a right track.
It is in our opinion that $g\mathbb{VASS}^d$ offers a more insightful classification of the VASS reachability problem.
The relationship between $g\mathbb{VASS}^d$ and $\mathbb{VASS}^d$ is of course an issue calling for further study.
Is it that $g\mathbb{VASS}^d\le_C\mathbb{VASS}^d$ or even $g\mathbb{VASS}^d\le_K\mathbb{VASS}^d$?
Here $\le_K$ is the Karp reduction (polynomial time reduction) and $\le_C$ is the Cook reduction (polynomial time Turing reduction).
Can we derive better lower bound for $g\mathbb{VASS}^d$ than for $\mathbb{VASS}^d$?
We have not talked about $g\mathbb{VASS}^1$ and $g\mathbb{VASS}^2$.
Is $g\mathbb{VASS}^1$ in $\mathbf{NP}$?
Is $g\mathbb{VASS}^2$ in $\mathbf{PSPACE}$?
An immediate future work is to derive completeness result for $g\mathbb{VASS}^1$ and $g\mathbb{VASS}^2$.

As far as complexity theoretical studies of $\mathbb{VASS}^d$ are concerned, the KLMST algorithm seems indispensable.
All simplifications of the algorithm are welcome.
The present paper provides one significant simplification. 
It is our hope that the version of the KLMST algorithm with its analysis tools presented here may be of help in the studies of BVASS and PVASS.
The reachability of these variants of VASS remains mysterious even in low dimension.

%¡ª¡ª¡ª¡ª¡ª¡ª¡ª¡ª¡ª¡ª¡ª¡ª¡ª¡ª¡ª¡ª¡ª¡ª¡ª¡ª¡ª¡ª¡ª¡ª¡ª¡ª¡ª¡ª¡ª¡ª¡ª¡ª¡ª¡ª¡ª¡ª¡ª¡ª¡ª¡ª¡ª¡ª¡ª¡ª¡ª¡ª¡ª¡ª¡ª¡ª

%\acks

%We thank W. Chen, H. Long, H. Wu and Q. Yin for proof-reading.

% We recommend abbrvnat bibliography style.

\bibliographystyle{abbrvnat}

% The bibliography should be embedded for final submission.

\end{document}